\documentclass{article}

\usepackage{arxiv}

\usepackage[utf8]{inputenc} 
\usepackage[T1]{fontenc}    
\usepackage{hyperref}       
\usepackage{url}            
\usepackage{booktabs}       
\usepackage{amsfonts}       
\usepackage{nicefrac}       
\usepackage{microtype}      
\usepackage{lipsum}		
\usepackage{graphicx}
\usepackage[round]{natbib}
\usepackage{doi}

\usepackage{tabularx,colortbl}

\usepackage{amsmath} 
\usepackage{amssymb}  

\usepackage[percent]{overpic}
\usepackage{xcolor}
\usepackage{rotating}
\usepackage{cases}

\allowdisplaybreaks 

\def\aujour{\number\day \space \ifcase\month\or
janvier\or f�vrier\or mars\or avril\or mai\or
juin\or juillet\or ao�t\or septembre\or octobre\or
novembre\or d�cembre\fi \space \number\year}

\def\cH{{\cal H}}

\def\cL{{\cal L}}

\newtheorem{thm}{Theorem}
\newtheorem{proof}{Proof}
\def\C{{\setbox0=\hbox{$\displaystyle{\rm C}$}
        \hbox{\hbox to0pt{\kern 0.4\wd0\vrule height 0.95\ht0\hss}\box0}}}
\def\Q{{\setbox0=\hbox{$\displaystyle{\rm Q}$}%
    \hbox{\raise 0.2\ht0\hbox to0pt{\kern 0.4\wd0\vrule height
    0.85\ht0\hss}\box0}}} 
\def\R{\mathop{\rm I\mkern -3.5mu R}} 
\def\N{\mathop{\rm I\mkern -3.5mu N}} 

\def\cHi{{\cal H}_{\infty}} 
\def\cH2{{\cal H}_2} 

\def\cL2{\mathop{\mathcal L}_{2}} 

\def\cRH2{\mathop{\cal R \cal H}_2} 
\def\cRL2{\mathop{\cal R \cal L}_{2}} 

\DeclareMathOperator*{\diag}{diag}

\DeclareMathOperator*{\der}{d}

\newcommand{\norm}[1]{\left\|{#1}\right\|}

\makeatletter
\DeclareRobustCommand\sfrac[1]{\@ifnextchar/{\@sfrac{#1}}
                                            {\@sfrac{#1}/}}
\def\@sfrac#1/#2{\leavevmode\kern.1em\raise.5ex
         \hbox{$\m@th\fontsize\sf@size\z@
                           \selectfont#1$}\kern-.1em
         /\kern-.15em\lower.25ex
          \hbox{$\m@th\fontsize\sf@size\z@
                            \selectfont#2$}}
\makeatother

\usepackage{cite}

\usepackage{siunitx}

\author{Ibrahima N'Doye,}
\author{Wenqi Cai,} 
\author{Asem Alalwan,}
\author{Xiaobin Sun,}
\author{Wael Ghazy Headary,}
\author{Mohamed-Slim Alouini,}
\author{Boon S. Ooi,}
\author{Taous-Meriem Laleg-Kirati} 

\title{Reduction of the Beam Pointing Error for Improved Free-Space Optical Communication Link Performance}

\author{ \href{https://orcid.org/0000-0002-2576-5515}{\includegraphics[scale=0.06]{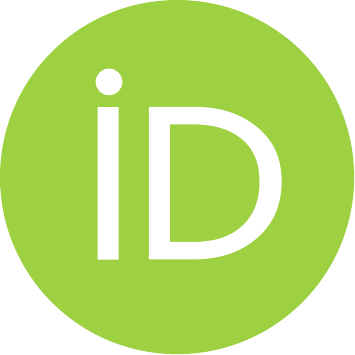}\hspace{1mm}I. N'Doye}, 
{W. Cai}, 
{A. Alalwan}, 
{X. Sun}, 
{W. G. Headary}, 
\href{https://orcid.org/0000-0003-4827-1793}{\includegraphics[scale=0.06]{orcid.pdf}\hspace{1mm}M.-S. Alouini}, 
{B. S. Ooi}, 
\href{https://orcid.org/0000-0001-5944-0121}{\includegraphics[scale=0.06]{orcid.pdf}\hspace{1mm}T.-M. Laleg-Kirati}
\thanks{This work has been supported by the King Abdullah University of Science and Technology (KAUST), Base Research Fund (BAS/1/1627-01-01) to Taous Meriem Laleg.} \\
Computer, Electrical and Mathematical Sciences and Engineering Division (CEMSE)\\
King Abdullah University of Science and Technology (KAUST)\\
Thuwal 23955-6900, Saudi Arabia \\
	\texttt{ibrahima.ndoye@kaust.edu.sa; slim.alouini@kaust.edu.sa; taousmeriem.laleg@kaust.edu.sa} \\
}

\renewcommand{\shorttitle}{\textit{arXiv} Template}

\begin{document}
\maketitle

\begin{abstract}
Free-space optical communication is emerging as a low-power, low-cost, and high data rate alternative to radio-frequency communication in short- to medium-range applications. However, it requires a close-to-line-of-sight link between the transmitter and the receiver. This paper proposes a robust $\cHi$ control law for free-space optical (FSO) beam pointing error systems under controlled weak turbulence conditions. The objective is to maintain the transmitter-receiver line, which means the center of the optical beam as close as possible to the center of the receiving aperture within a prescribed disturbance attenuation level.  First, we derive an augmented nonlinear discrete-time model for pointing error loss due to misalignment caused by weak atmospheric turbulence. We then investigate the $\cHi$-norm optimization problem that guarantees the closed-loop pointing error is stable and ensures the prescribed weak disturbance attenuation. Furthermore, we evaluate the closed-loop outage probability error and bit error rate (BER) that quantify the free-space optical communication performance in fading channels. Finally, the paper concludes with a numerical simulation of the proposed approach to the FSO link's error performance.
\end{abstract}


\section{Introduction}
Free-space optical (FSO) communication systems have emerged as a viable technology that offers a large capacity usage (data, voice, and video) in short to medium-range applications. The range of applications include fixed-location terrestrial communication \citep{WeG:01}, communication between mobile robots \citep{KBGW:96}, underwater wireless optical communication (UWOC) \citep{OLPKAO:15}, airborne communication \citep{May:87}, and inter-satellite communication \citep{Cha:03}. FSO is a line of sight communication network with a free-space or atmosphere as a channel. This channel may be turbulent, causing absorption and scattering of the optical signal due to the presence of many factors, including fog, rain, snow, and temperature variations, resulting in its deterioration \citep{SiS:12}. Due to the temperature variations in the atmosphere, the refractive index changes creating Fresnel zones of different densities that scatter the laser beam from its projected path to travel diverse directions.

Atmospheric turbulence is a random phenomenon caused by the variation of temperature or humidity and the atmosphere's pressure along the propagation path. Specifically, atmospheric turbulence is highly variable and unpredictable due to weather effects \citep{HeW:10}. It would make the optical beam fluctuated when propagating through the channel and finally results in misalignment due to diffraction from particulates present in the channel, resulting in enlarging the beam's size to become more significant than the receiver aperture size. Misalignment can lead to intolerable signal fades and can significantly degrade system performance. In other words, atmospheric turbulence may lead to a significant degradation in the performance of the FSO communication systems \citep{APA:17}.

In addition to this, the signal propagating through the FSO channel is also perturbed by building vibrations, sways, and thermal expansions result in degradation of link performance \citep{ShC:02}, \citep{FaH:07}. This misalignment can lead to pointing errors, causing the optical beam's displacement along with horizontal and vertical directions. Hence, FSO links require accurate pointing \citep{BoV:09}, which means the pointing error needs to be very small to reduce the loss due to misalignment between the transmitter-receiver line.

Since FSO systems require precise pointing as the light signals are highly directional, the effect of pointing errors on link performance is a great interest for many potential applications. Several approaches have been proposed to address the LOS (Line-of-Sight) requirement in optical communication systems. In \citep{PFWPP:08}, large-area photomultiplier tubes are used to increase the receiver's field of view. Multiple LEDs and multiple photodiodes have also been used to avoid the need for active pointing during optical-communication \citep{RuA:12}, \citep{SHM:12}. However, these systems achieved the LOS through redundancy in transmitters and receivers, which resulted in a larger footprint, higher cost, and higher complexity.

Furthermore, different pointing strategies for FSO links have been proposed in \citep{ARK:97,KKN:07,YMD:05,Liu:09b,CNOAL:19}. However, these methods focus on combining existing components of the pointing assembly and atmospheric turbulence effects using manual and special detection techniques. Additionally, these methods use statistical performance analysis tools to mitigate pointing error effects but did not consider the controller design aspect. Many of them did not include the influence of vibration levels and atmospheric turbulence. To alleviate these shortcomings and arrive with an accurate pointing error solution in FSO links, we propose a robust control strategy for maintaining the optical link between free-space communication stations engaged in a laser communication channel.

To the best of the authors' knowledge, designing a beam pointing error control for improved FSO link performance has not been thoroughly investigated. One of this paper's motivations is to study and characterize the lognormal turbulence fading theoretically and experimentally to construct fully auxiliary control subsystems for robust FSO links.

The contributions of this paper to the existing body in the literature in pointing error control for FSO systems are as follows.
\begin{itemize}
  \item We propose an experimental setup and analyze the lognormal fading of the weak turbulence FSO channel. 
  \item We derive a discrete-time nonlinear model based on a predefined autocorrelation function using the implicit Milstein scheme to simulate the lognormal optical channel state.
  \item Using the discrete-time model, we propose the $\cHi$-norm optimization problem that guarantees the closed-loop pointing error is stable and ensures a prescribed disturbance attenuation level.
  \item We evaluate the quality of the closed-loop pointing error control, which shows that the proposed control law can maintain the optical beam's center at the center of the receiving aperture.
  \item Finally, we perform numerical simulation tests of the open-loop and closed-loop outage probability error and bit error rate (BER) that quantify the free-space optical communication's performance in fading channels.
 \end{itemize}

This paper is an extension of \citep{CNSAAAOL:18}, with the following significant new contributions.
\begin{itemize}
  \item  A discrete-time nonlinear model based on a predefined autocorrelation function is adapted to capture the lognormal fading process and provide a comprehensive treatment of the optical beam model. Indeed, the lognormal random process is represented as a solution of a stochastic differential equation (SDE), which is approximated by a general and effective discrete-time model.
  \item A robust $\cHi$ control law is proposed to reduce the pointing error and maintain the line-of-sight link of the optical beam under controlled lognormal weak-turbulence conditions. 
  \item The communication performance metrics using the outage probability error and the bit error rate (BER) are evaluated and analyzed.
\end{itemize}

The outline of the paper is organized as follows: In Section \ref{sec-WeakT}, we describe the experimental setup of the FSO channel in which we characterize the lognormal fading and derive the discrete-time lognormal optical channel state. In Section \ref{sec-PF}, we formulate the pointing error problem to maintain the centroid of the optical beam as close as possible to the center of the photodetector. In Section \ref{main}, the main results of the pointing error problem based on the $\cHi$-norm optimization method under controlled weak-turbulence conditions are derived. In Section \ref{results}, we evaluate the quality of the pointing error control and the communication performance metrics through numerical simulation tests. Finally, concluding remarks of the proposed robust pointing error control are presented in Section \ref{conclusion}.

\noindent {\bf Notations.} 
$M^{T}$ is the transpose of $M$. In symmetric block matrices, the symbol $(\star)$ in any matrix represents for any element that is induced by transposition. $\displaystyle  \norm{.}$ is the induced $2$-norm. 
$\mathbf{0}$ and $\mathbf{I}$ stand for the null matrix and the identity matrix of appropriate dimensions, respectively. 

\section{Experimental Setup and Dynamic Model for Weak Turbulence FSO Channel}\label{sec-WeakT}
The FSO link consists of a transmitter and receiver separated by the atmospheric channel. Here we set up an experimental system, the schematic diagram of the experimental line-of-sight FSO link is shown in Fig. \ref{fig-Block}. The optical signal amplitude through the FSO channel is fluctuated due to the atmospheric turbulence. Many statistical models of the intensity fluctuation through FSO channels have been proposed in the literature for distinct turbulence regimes.  For weak turbulence conditions, the most widely used model is the log-normal distribution, which has been validated through studies \citep{FaH:07}, \citep{MaR:08}, \citep{ZhK:02}. It is a well-known modeling approach and has been adopted in many calculations for the turbulence channel. This paper will focus on the weak turbulence; therefore, the lognormal model will be used throughout.
\begin{figure}[!t] 
        \centering
     \begin{overpic}[scale=0.51]{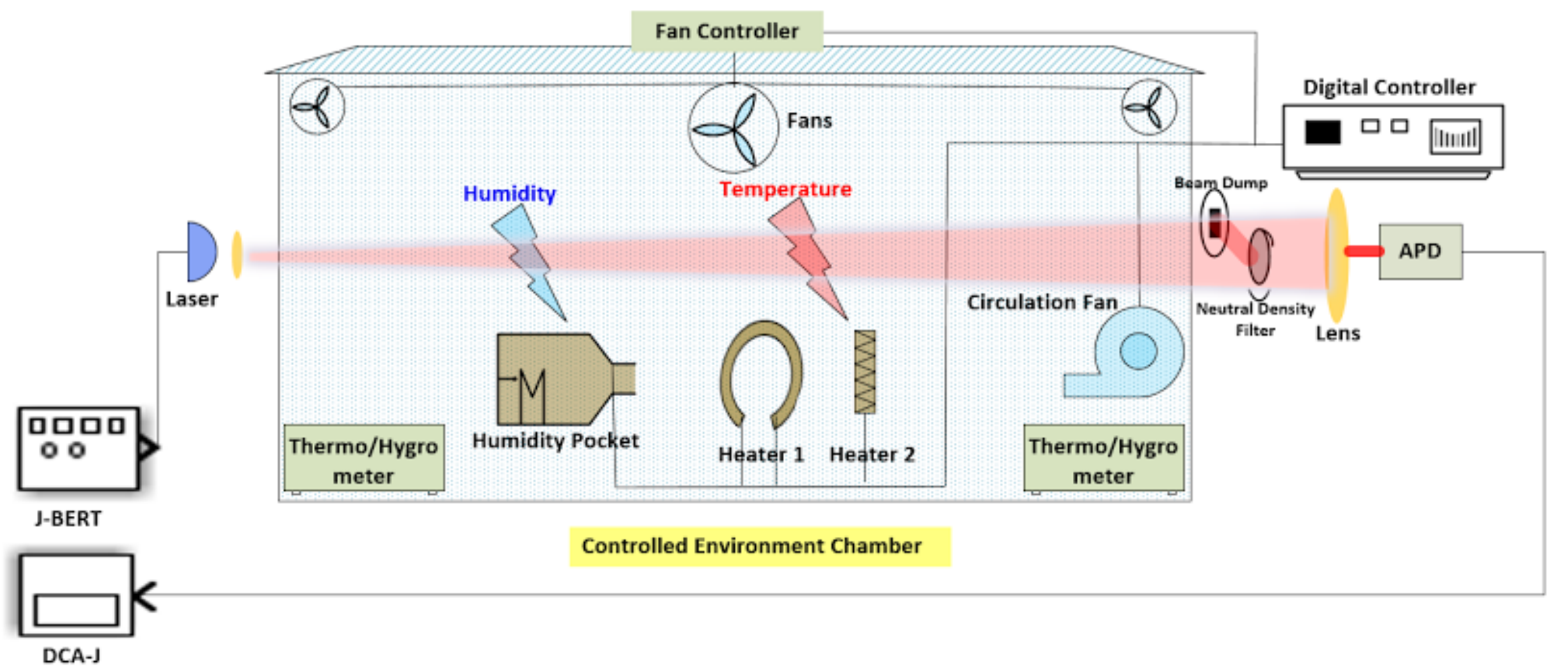}
            \end{overpic}  \vspace{-8pt}
              \caption{Block diagram of the experiment FSO setup.}\label{fig-Block}  
\end{figure}

\subsection{Lognormal Weak Turbulence Characterization}\label{sec-Log}
The laboratory atmospheric channel is a closed glass chamber with a dimension of $100\times35\times42 \mbox{cm}^3$ as depicted in Fig. \ref{fig-Chamber} with the aim of observing the effect of atmospheric turbulence on the laser beam propagating through the channel. The main parameters of the FSO link are given in Table \ref{FSO-para}.
\begin{figure}[!t] 
        \centering
     \begin{overpic}[scale=0.45]{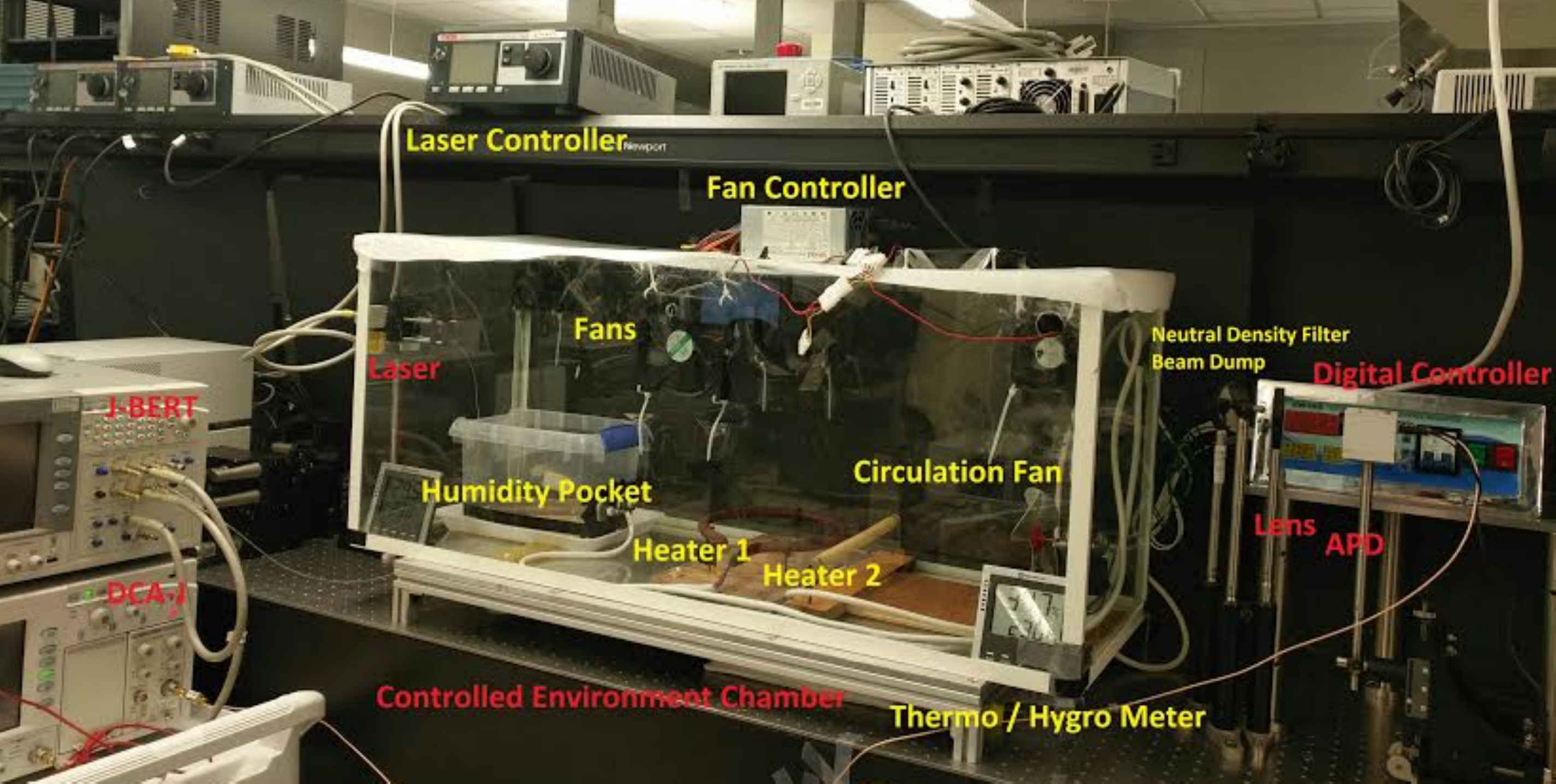}
            \end{overpic}  \vspace{-4pt}
              \caption{Experimental laboratory turbulence chamber.}\label{fig-Chamber}                
\end{figure}  
\begin{table}[!t]
\begin{normalsize}
\caption{Parameters of the FSO link.}
\begin{center}
\begin{tabular}{| c || c || c |}
  \hline \rowcolor{gray}
\textbf{Description}  & \textbf{Parameter}  & \textbf{Value}   \\ \hline  
Data &  Format & OOK NRZ \\  
~~~~ &PRBS length&$2^{10}-1$\\  
~~~~~~~~~ &Signal intensity $V_{\mbox{\scriptsize peak-to-peak}}$ &$1.78$V\\  
~~~~~~~~~ &~~~Data rate~~~ &$622.082$Mb/s \\  \hline
Laser diode &  Type  &LP642-SF20  \\  
~~~~~~~~~ &Peak wavelength &$642$nm\\  
~~~~~~~~~ &~~Optical Output Power  &$20$mW\\  
~~~~~~~~~ & Operating current/voltage &$0.089$A/$2.371$V\\  \hline
~~Photodetector~~ &  Type & APD210  \\  
~~~~~~~~~ &~Spectral range~ &$400 -1000$nm\\  
~~~~~~~~~ &~Maximum gain~ &$2.5\times10^5$V/W\\  
~~~~~~~~~ &~Detector diameter~ &$0.5$mm\\  
~~~~~~~~~ &~Rise time~ &$0.5$ns\\  \hline
Lens &  Type &LA1417-A  \\  
~~~~~~~~~ &Diameter&$50.8$mm\\  
~~~~~~~~~ &~~Focal length~ &$150$mm\\  \hline
Transmitter &  Type &~N4903B J-BERT~  \\  \hline
Receiver &  Type &86100C-DCA-J  \\  
~~~~~~~~~ &Sampling time&$0.2$ns\\  \hline
Chamber &  Dimension &$100\!\times\!35\!\times42 \mbox{cm}^3$  \\  \hline
\end{tabular}
\end{center}
\label{FSO-para}
\end{normalsize}
\end{table}
The probability density function (PDF) of the received irradiance $I$ due to the turbulence is derived by \citep{Osc:02}, \citep{GLRPI:12}
\begin{equation}\label{eq-aa}
p(I)=\frac{1}{\sqrt{2\pi\sigma^2}}\frac{1}{I}\displaystyle\exp\left\{-\frac{(\ln(I/I_0)+\sigma^2/2)^2}{2\sigma^2}\right\}
\end{equation}
where $I_0$ is the irradiance when there is no turbulence and $\sigma^2$ is the log-amplitude variance or scintillation index in the channel. 

The measured eye-diagrams for the received signal are depicted in Fig. \ref{fig-NoT} without turbulence and in Figs. \ref{fig-Turba} and \ref{fig-Turbb} with weak turbulence. We observe that the eye-opening is smaller in the presence of turbulence, which results in a considerable level of signal intensity fluctuation and also reduces the FSO performance link.
\begin{figure*}[!t]
   \begin{minipage}[c]{0.39\linewidth}  
        \centering
     \begin{overpic}[scale=0.36]{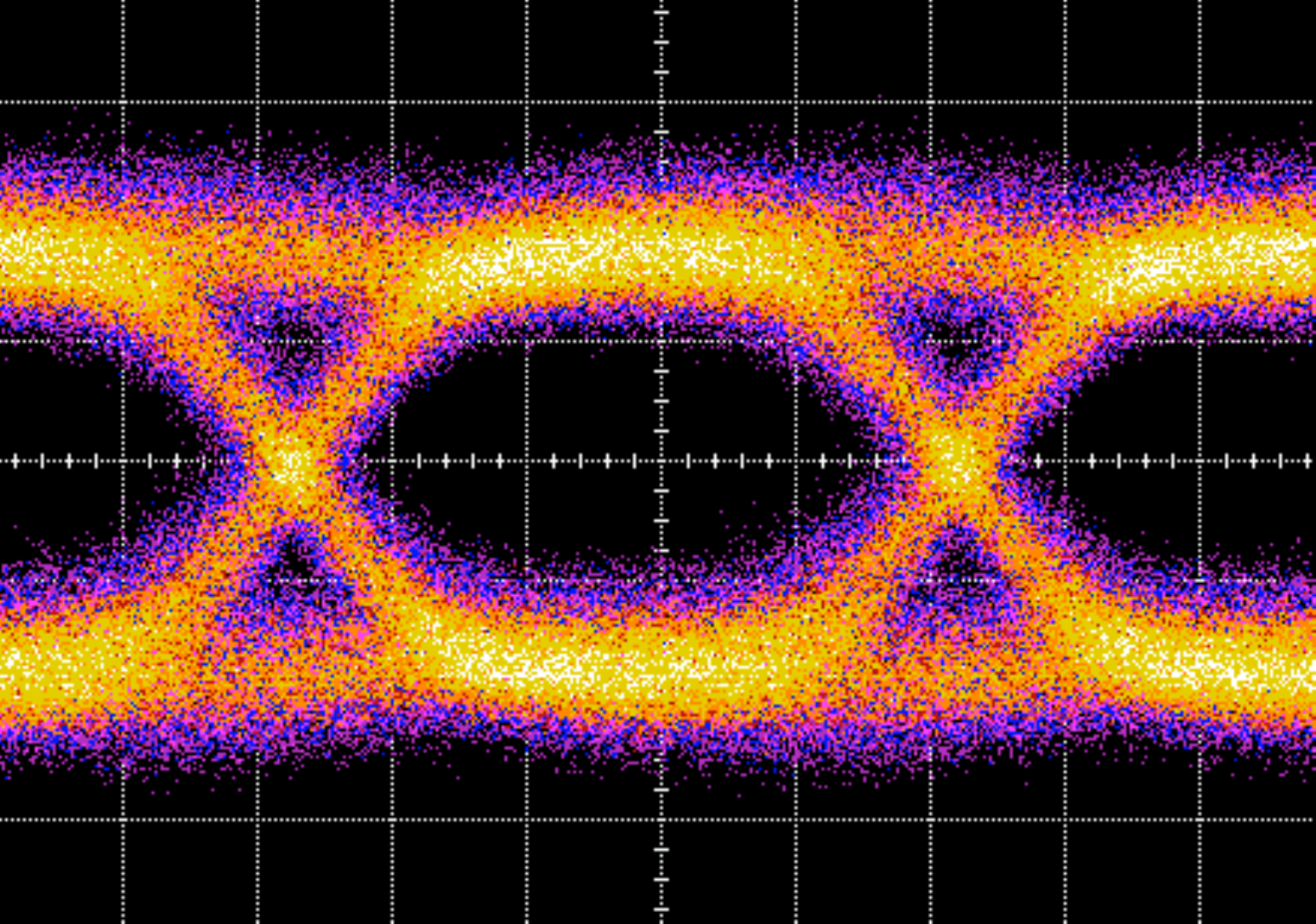}
            \end{overpic}  \vspace{2pt}
              \caption{Measured screen shot eye-diagram of received intensity signal without turbulence}\label{fig-NoT}
       \end{minipage}\hfill 
         \begin{minipage}[c]{0.39\linewidth}
     \begin{overpic}[scale=0.36]{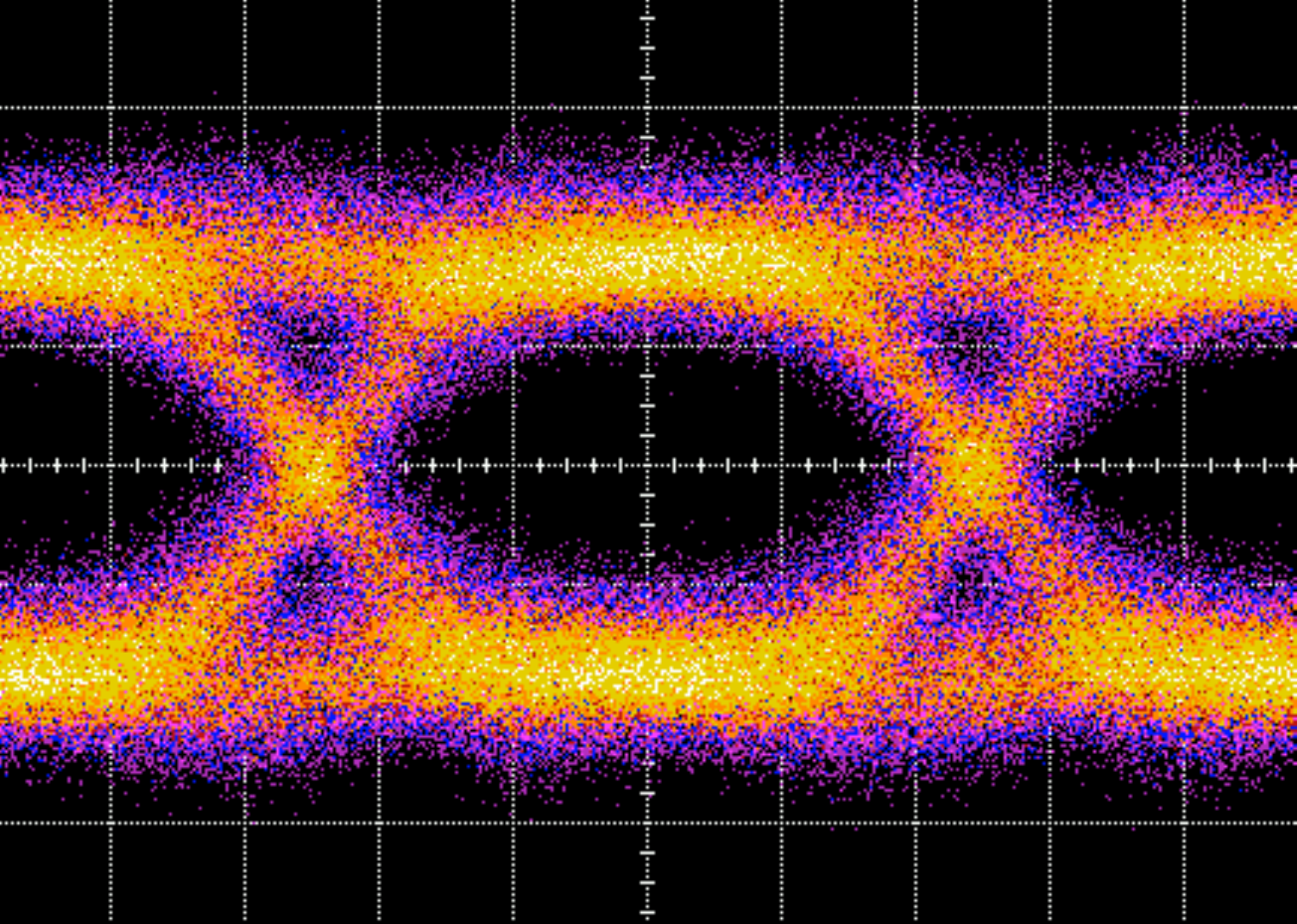}
            \end{overpic}  \vspace{2pt}
              \caption{Measured screen shot eye-diagram of received intensity signal under weak turbulence: $\sigma^2=0.0380$.}\label{fig-Turba}
   \end{minipage} 
  \end{figure*} 
 \begin{figure*}[!t]
\centering
     \begin{overpic}[scale=0.36]{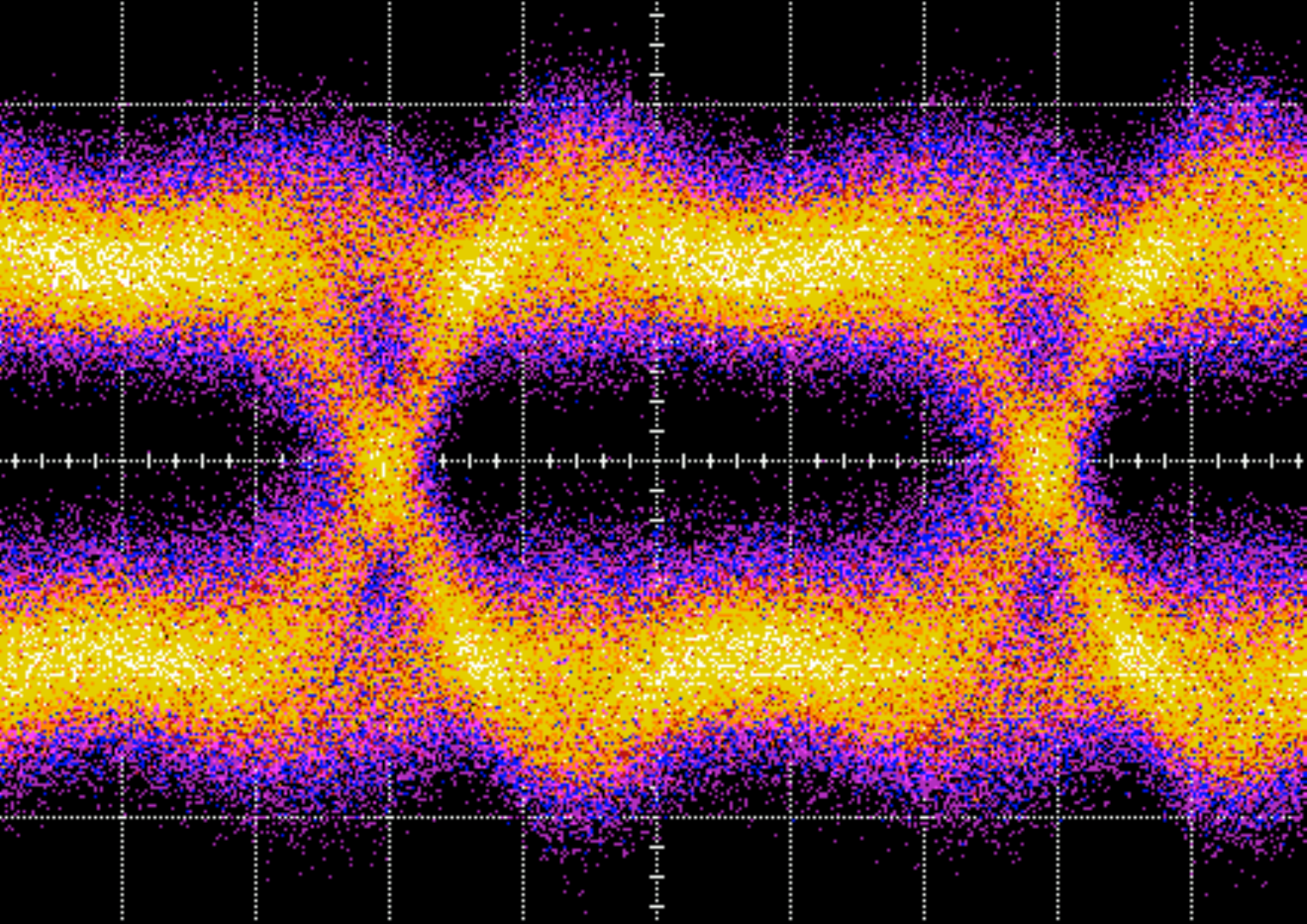}
            \end{overpic}  \vspace{2pt}
              \caption{Measured screen shot eye-diagram of received intensity signal under weak turbulence: $\sigma^2=0.0576$.}\label{fig-Turbb}  
      \vspace{0.5cm} 
\end{figure*}
 
 \begin{figure*}[!t]
\centering
      \begin{overpic}[scale=0.26]{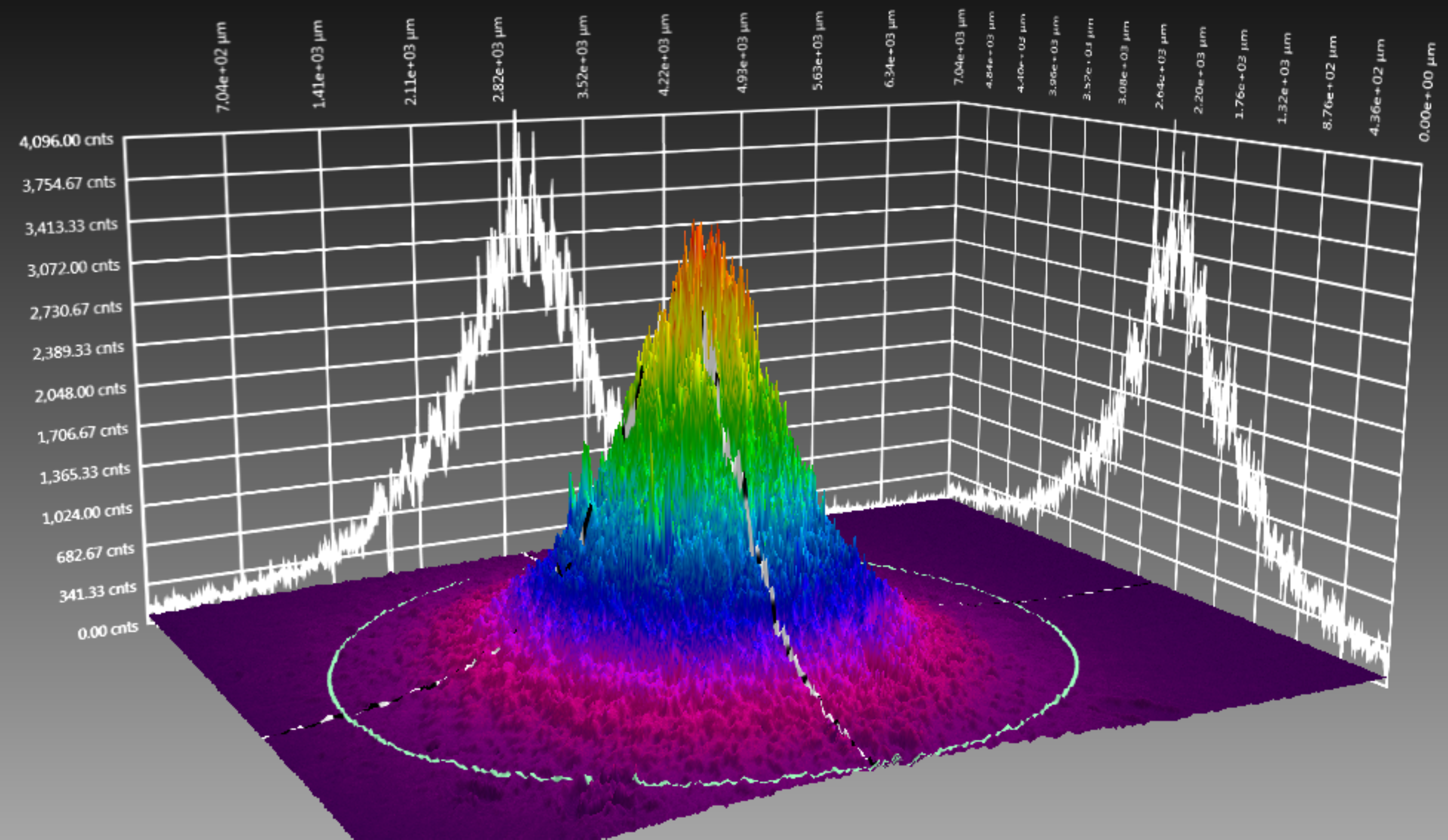}
         \put(-1.5,22){\scriptsize  \begin{rotate}{90} Intensity [mv]\end{rotate}}
         \put(6,5){\footnotesize  $x$[mm]}
            \put(64,3){\footnotesize  $y$[mm]}
            \end{overpic}  \vspace{2pt}
              \caption{Experimental snapshot intensity of the laser beam obeying to Gaussian distribution.}\label{fig-3D}    
\end{figure*}

\begin{figure*}[!t]
         \vspace{0.5cm}
   \begin{minipage}[c]{0.40\linewidth}  
        \centering
     \begin{overpic}[scale=0.39]{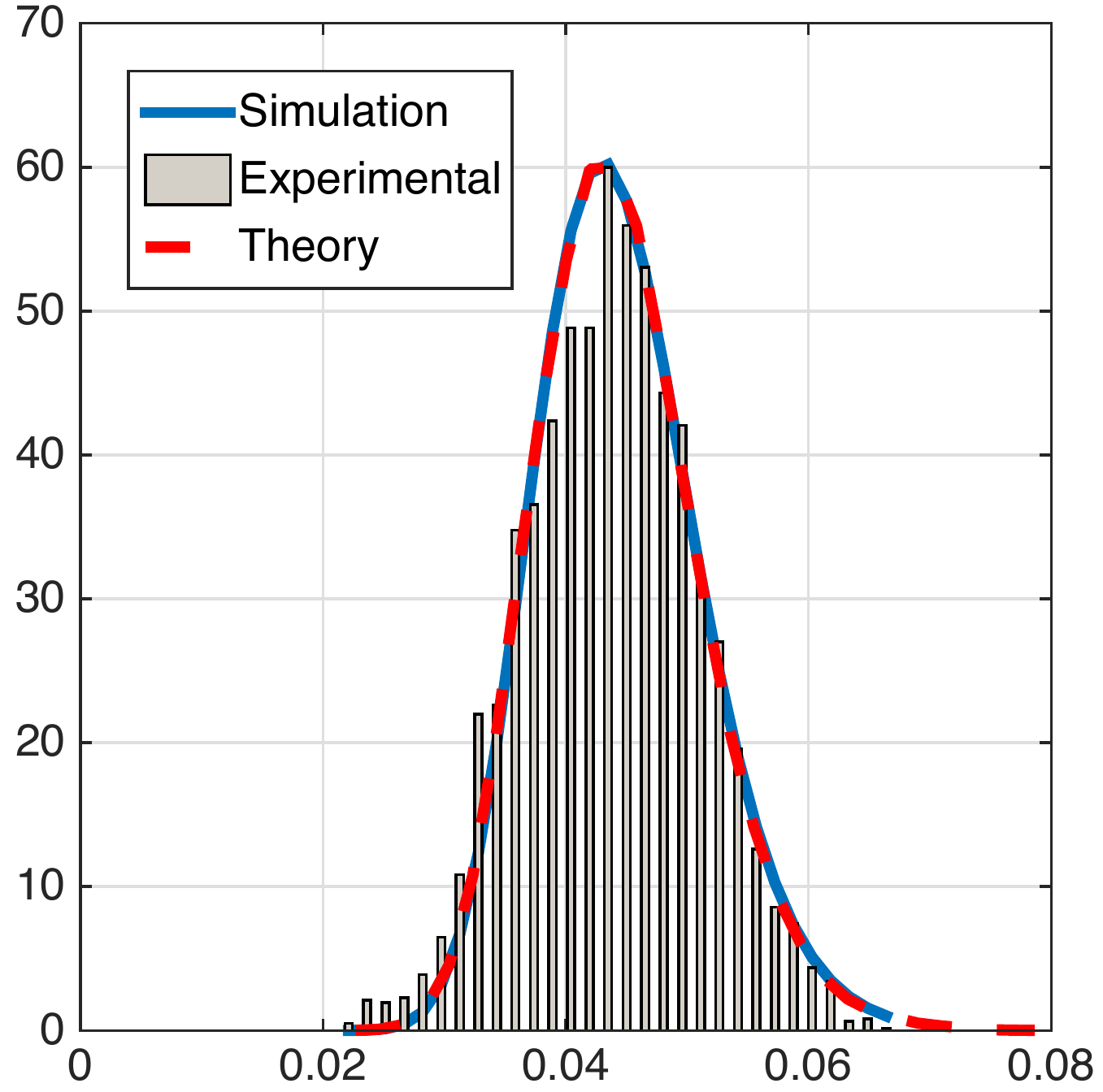}
   \put(-2,20){\scriptsize  \begin{rotate}{90} Probability density \end{rotate}}
 \put(33,-5){\scriptsize  Irradiance (mv)}
            \end{overpic}  \vspace{2pt}
              \caption{Gaussian PDF received distribution without turbulence (the curve fitting is shown by solid lines)}\label{fig-PDFno}
       \end{minipage}\hfill 
         \begin{minipage}[c]{0.40\linewidth}
     \begin{overpic}[scale=0.38]{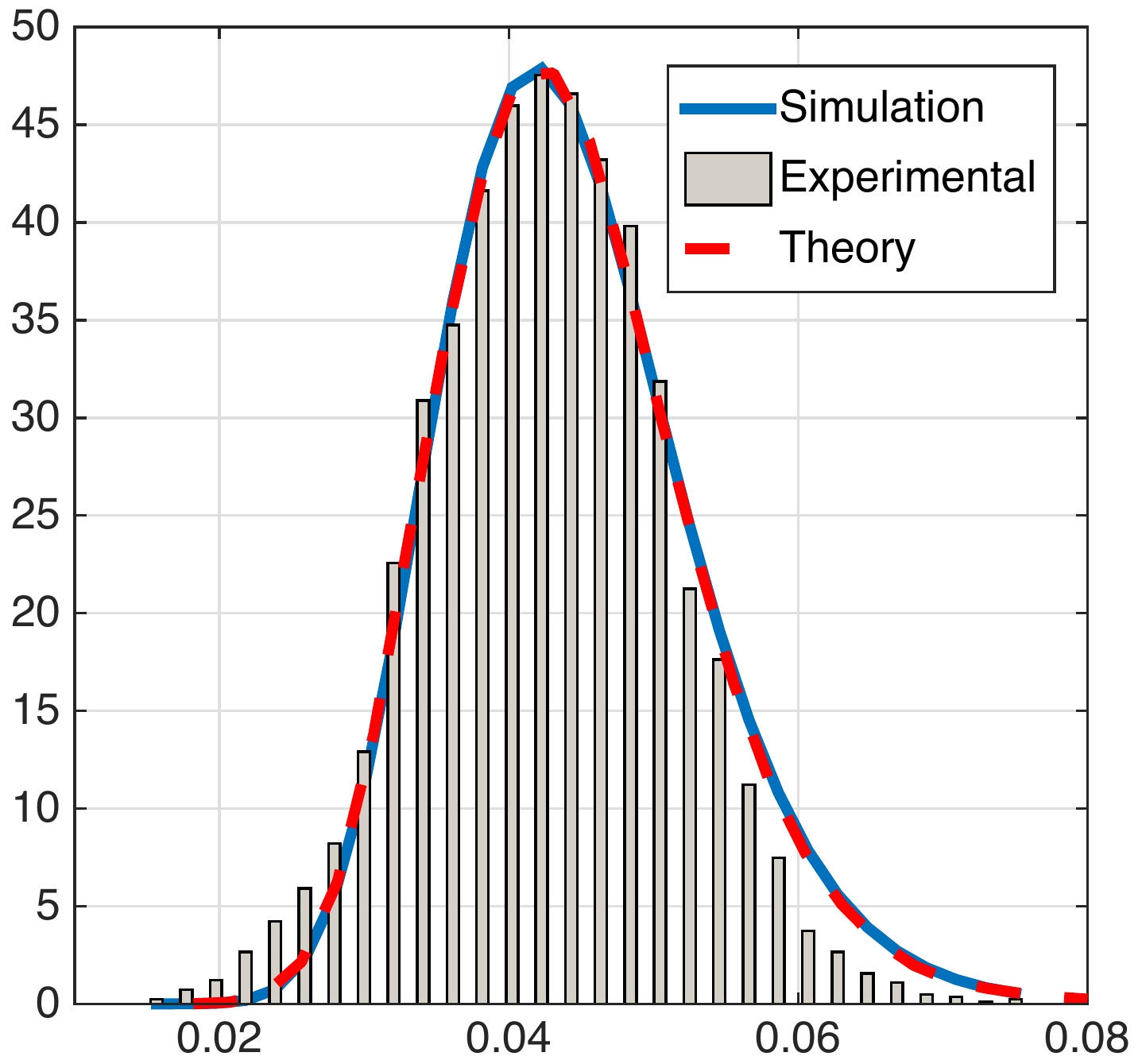}
   \put(-2,20){\scriptsize  \begin{rotate}{90} Probability density \end{rotate}}
 \put(33,-5){\scriptsize  Irradiance (mv)}
            \end{overpic}  \vspace{2pt}
              \caption{Log-normal PDF received distribution under weak turbulence: $\sigma^{2}\!=\!0.0380$ (the curve fitting is shown
by solid lines).}\label{fig-PDFTurb}
   \end{minipage} 
 \end{figure*}
   
\begin{figure*}[!t]
\centering
     \begin{overpic}[scale=0.38]{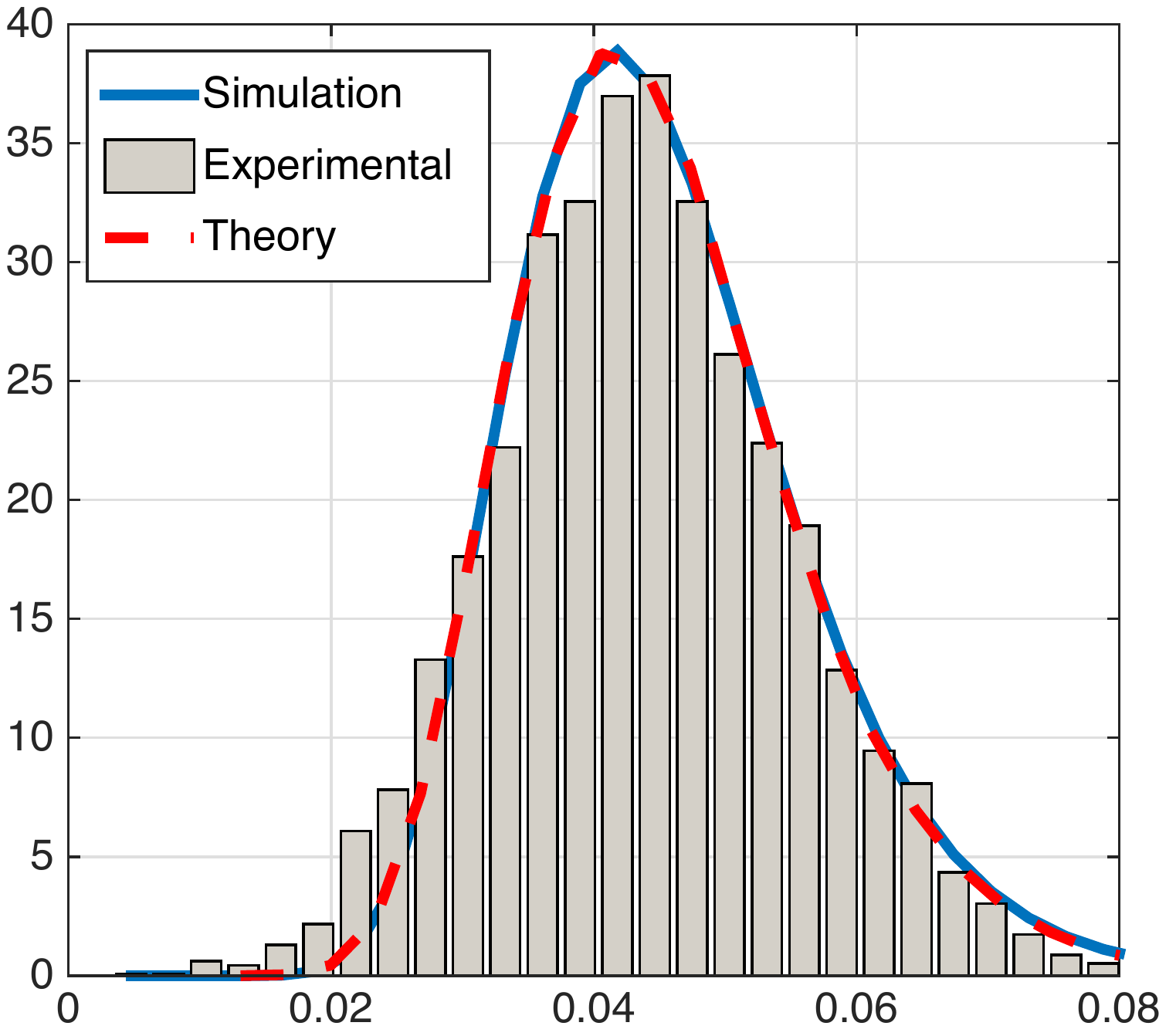}
   \put(-2,20){\scriptsize  \begin{rotate}{90} Probability density \end{rotate}}
 \put(33,-5){\scriptsize  Irradiance (mv)}
            \end{overpic}  \vspace{2pt}
\caption{Log-normal PDF received distribution under weak turbulence: $\sigma^{2}\!=\!0.0576$ (the curve fitting is shown
by solid lines).}\label{Turb_fig2}  
         \vspace{0.5cm}   
\end{figure*}

Fig. \ref{fig-PDFno} shows the histogram and the curve fitting plots of the received intensity signal without turbulence. As we can see, the PDF distribution is nearly Gaussian for lower values of $\sigma{^2}$. The experimental snapshot intensity sensed by the PD can be fitted with a nearly Gaussian distribution, as illustrated in  Fig. \ref{fig-3D}.

Figs. \ref{fig-PDFTurb} and \ref{Turb_fig2} illustrate the histograms and the curve fitting plots of the received intensity signal with turbulence. It is clear that the PDF has a good fitting with the lognormal distribution, and the estimated scintillation index $\sigma^{2}$ falls within the range of $[0,\,0.1]$ is characterized by weak turbulence regime \citep{GLRPI:12}. The results also show the theoretical red-dotted-lines fit well with the simulation solid-blue-lines, which demonstrate the close resemblance between the PDFs fading statistics based on lognormal distributions of the theoretical predefined autocorrelation function (see \citep{PKL:05,KoL:95,BEA:15,NeS:14} and the simulation for short-range turbulent-channel communication experiment. As $\sigma^{2}$ increases, the distribution is more skewed with a long tail toward the infinity and reduced peak of probability density as a result of signal fading.

 \subsection{Discrete-Time Lognormal Optical Channel State}\label{sec-Mod}
The most important property of the optical beam is the PDF of gain samples, so we use a modeling approach to get the optical beam position. The modeling approach is based on 1D lognormal distributed samples with a corresponding correlation function, as illustrated in Figs. \ref{fig-PDFno}, \ref{fig-PDFTurb} and \ref{Turb_fig2}. The first step is to efficiently approximate the weak turbulence level of the FSO chamber, which means to emulate the variance of the atmospheric turbulence. We verify the theoretical results of the lognormal process with the predefined autocorrelation function using the implicit Milstein scheme for the channel states, which converges to a simple discrete-time differential equation \citep{NeS:14}, \citep{Byk:15}. The implicit discrete-time Milstein scheme for the lognormal distribution describing the simulated lognormal optical beam channel state and its relative position is generated by the following nonlinear discrete-time state-space equations \citep{NeS:14}, \citep{KoL:95}, \citep{Byk:15} 
\begin{equation}\label{eq-1}
\left\{\begin{array}{llll}
x_{k+1}^p=a^px_{k}^p+\varphi(x_{k}^p)+b^pu_{k}^p + r^p w_{k}^p,\\
\theta_k=c^px_{k}^p,
\end{array}\right.
\end{equation}
where
\begin{equation}\label{autocorrelation-eq-2}
\varphi(x^p_{k})=-\frac{K}{2\sigma^2x^p_k}\left[\ln(x^p_k/I_0)\right], \quad r^p=\sqrt{K\times \Delta t},
\end{equation}
$K$ is given by
\begin{equation}\label{autocorrelation-eq-3}
K=\frac{2I_0^2\exp(\sigma^2)[\exp(\sigma^2)-1]}{\tau_c},
\end{equation}
where $k\in \mathbb{Z^+}$ is the set of all nonnegative integers, $x_{k}^p\in\R$ is the simulated optical channel state which can be  considered as a moving object, $\theta_k\in\R$ is the position of the optical beam transmitter,  $w_{k}^p\in\R$ are samples of the white Gaussian noise and $\tau_c$ is a predefined correlation time, $u_{k}^p$ is the bounded control input through which the optical channel state and transmission angle are changed. $\Delta t$ is the sampling time, $a^p\!\in\!\R$, $b^p\!\in\!\R$, $c^p\!\in\!\R$ and $r^p\!\in\!\R$ are constant values. The nonlinearity $\varphi(x_{k}^p)\!\in \mathcal{P}\subseteq \!\R^+$ represents here the full signal strength model. It is differentiable with $\varphi(0)=0$, locally Lipschitz and also monotonically increasing in  $\mathcal{P}$.
The time-correlated position of the transmitted optical beam that was generated by \eqref{eq-1} with $\sigma^2=0.0380$, $\tau_c=0.1$, $a^p=1$, $b^p=1$ and $c^p=1$ is illustrated in Fig. \ref{fig-beam-position}.

\begin{figure}[!t]
\centering
      \begin{overpic}[scale=0.36]{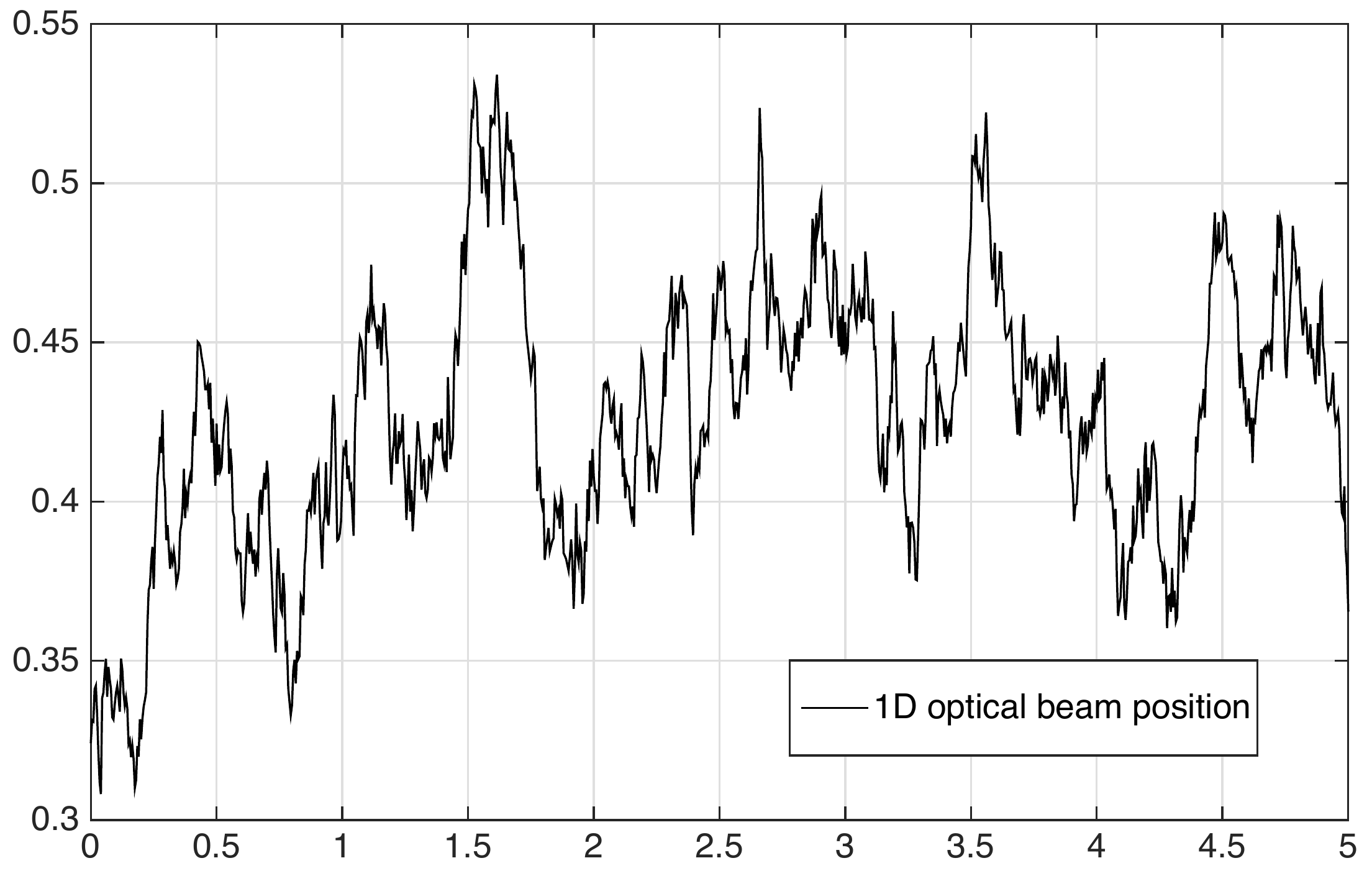}
         \put(-0.5,22){\scriptsize  \begin{rotate}{90} $\theta$-displacement [{\si mm}]\end{rotate}}
 \put(44,-2){\scriptsize  Time [{\si sec}]}
            \end{overpic}  \vspace{-2pt}
              \caption{Position of the transmitted optical beam motion versus time.}\label{fig-beam-position}    
\end{figure}

\subsection{Receiving Aperture Model}\label{Receiver}
Although the photodetector's receiver aperture is fixed, it still suffers some random physical vibrations due to thermal expansion, voltage jitter, etc (see \citep{CNOAL:19}). So, the receiving aperture motion is assumed similar to the Brownian motion of a particle subjected to excitation, as showed in Fig. \ref{fig-Brownian}. The Brownian motion is given by the generalized differential equation \citep{VoV:13,AFF:13,AFF:14}
\begin{equation}\label{RAperture-eq1}
m\frac{\der^2\! x(t)}{\der\!t^2}=-\gamma_1\frac{\der\! x(t)}{\der\! t}-k_1\frac{\der\! x(t)}{\der\! t}+\sqrt{2k_BT\gamma_1}W(t),
\end{equation}
where $x(t)$ is the trajectory of the particle with respect to the center, $m$ is the particle mass, $\gamma_1$ is the friction exerted by the surrounding medium on the particle, $k$ is the optical trap stiffness, $k_BT$ is the thermal energy unit, $k_B$ is the Boltzmann constant, $T$ is the absolute temperature and $W(t)$ is white Gaussian noise.
The discrete-time state space of the receiving aperture model is derived from the discretized particle Brownian motion \eqref{RAperture-eq1}. It is given as follows \citep{VoV:13,AFF:13,AFF:14}
\begin{equation}\label{RAperture-eq2}
\left\{\begin{array}{llll}
x_{k+1}^l= a^lx_{k}^l+r^lw_{k}^l,\\
\alpha_k=c^lx_{k}^l,
\end{array}\right.
\end{equation}
where $x_k^l$ is the source position, $\alpha_k$ is the measured position of the receiving aperture motion, $w_{k}^l$ is a standard white Gaussian noise, $\Delta T$ is the discretization time step, $a^l\!=\!\left(1-\frac{k_1\Delta T}{\gamma_1}\right)$, $r^l\!=\!\sqrt{2k_BT\gamma_1}$ and $c^l\!=\!1$. 
\begin{figure}[!t]
\centering
      \begin{overpic}[scale=0.35]{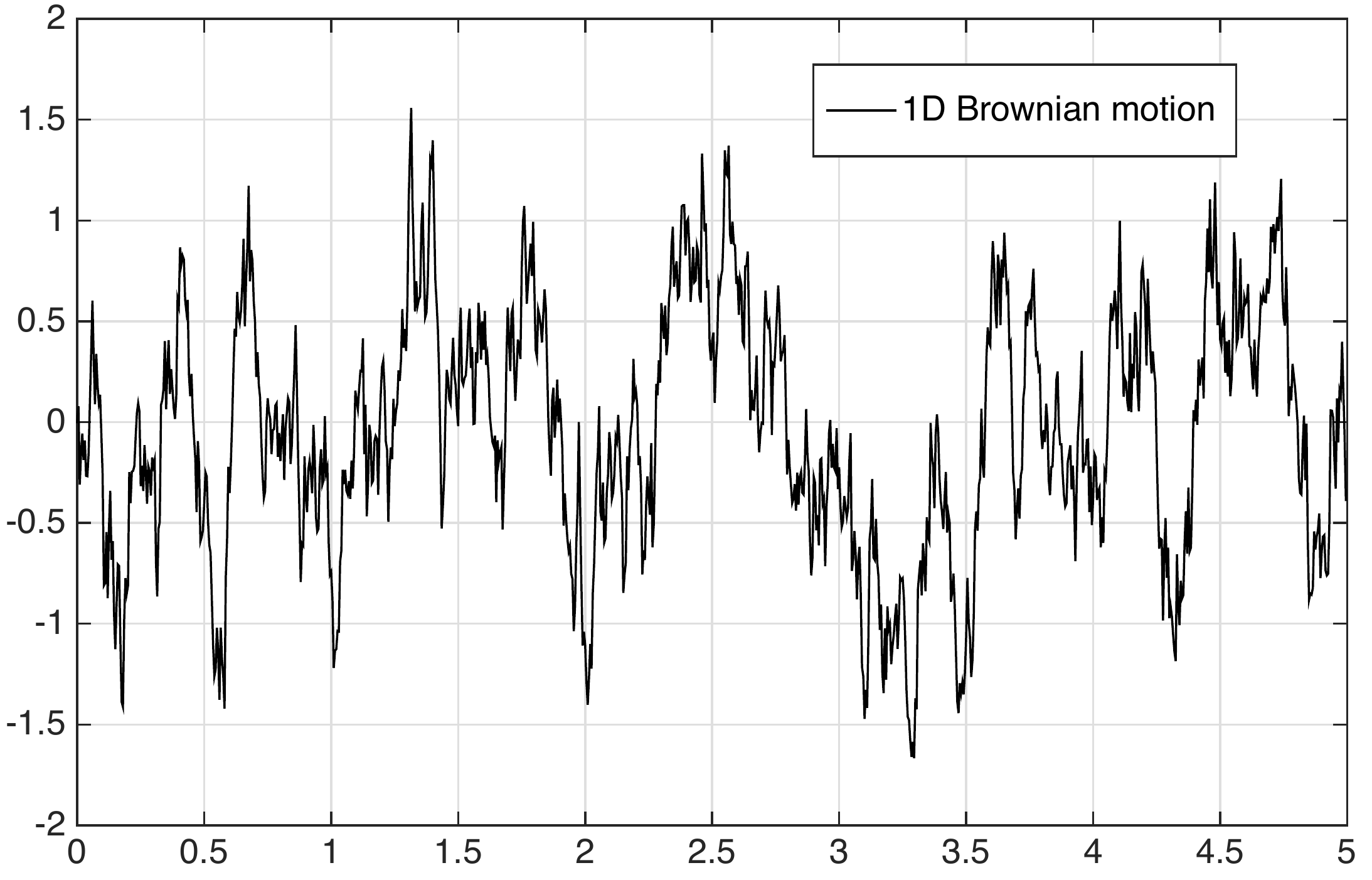}
         \put(-0.5,22){\scriptsize  \begin{rotate}{90} $\alpha$-displacement [{\si mm}]\end{rotate}}
 \put(44,-2){\scriptsize  Time [{\si sec}]}
            \end{overpic}  \vspace{2pt}
              \caption{Position of the 1D receiving aperture motion versus time.}\label{fig-Brownian}    
\end{figure}

\section{Problem Formulation}\label{sec-PF}
We consider a one-way optical link that consists of an optical transmitter and an optical receiver. Both are subject to relative motions. The emitted optical beam has a non-uniform intensity profile, which is assumed to be Gaussian \citep{GaK:95} and can be considered as a moving object. The goal is to control how the position of the object change in time. The receiver's aperture is assumed to be smaller than the received optical beam so that the receiver can collect only a fraction of the optical beam \citep{KKN:07}. This captured fraction can be enlarged by active pointing whose objective is to maintain the optical beam's center at the center of the receiving aperture. A photodetector is used at the receiver to measure the optical beam's intensity profile that strikes its aperture \citep{KKN:07}. The output is then sent as feedback through an optical link or low-bandwidth RF channel and used to adjust the transmitter's position. Fig. \ref{fig-1} illustrates the block diagram of this active pointing scheme under weak controlled turbulence.

\begin{figure}[!htbp]
\vspace{0.5cm}
\centering
      \begin{overpic}[scale=0.23]{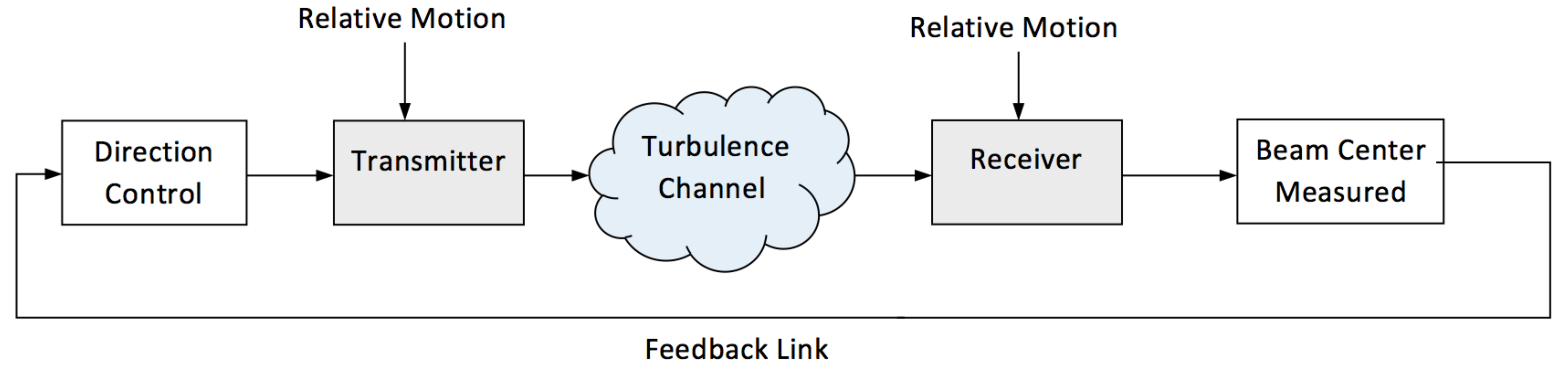}
            \end{overpic}  \vspace{-4pt}
              \caption{Active pointing scheme for a short range free-space optical channel.}\label{fig-1} 
              \vspace{0.5cm}   
\end{figure}

The discrete-time model considered in this study has been derived from the model structure that was introduced for the stochastic state-space model \citep{KKN:07}, \citep{Byk:15}. Indeed, the lognormal random process is represented as a solution of a stochastic differential equation (SDE), which is approximated by a general and effective discrete-time model. The model mainly describes the relative position's effect between the transmitter and the receiver on the signal strength. We denote the transmitted optical beam position to a fixed coordinate system by the vector  $\theta_k$ and the position of the receiving aperture of the stations to the same coordinate system by $\alpha_k$.
We assume that the receiving aperture is held perpendicular to the line-of-sight optical beam. The relative displacement of the optical beam center to the receiving aperture center is given by $y_k=d(\theta_k-\alpha_k)$ where $d$ is the distance between the transmitter and receiver. Fig. \ref{fig-2} illustrates the optical beam in the plane of the receiving aperture and the displacement vector $y_k$.

\begin{figure}[!t]
\centering
      \begin{overpic}[scale=0.66]{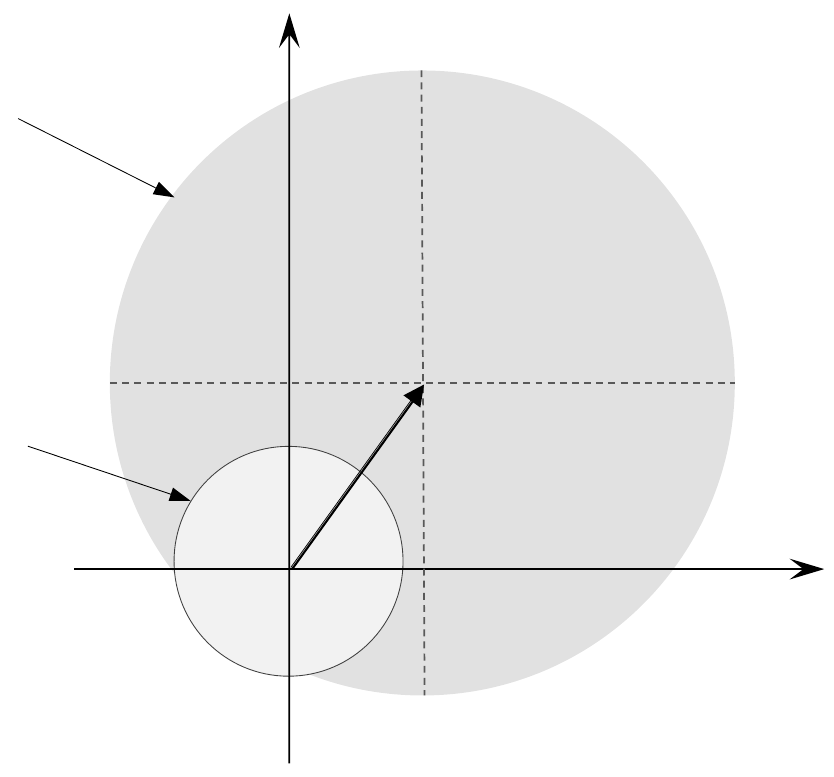}
       \put(45,35){\footnotesize  $y_k$}
        \put(-43,42){\footnotesize  Receiving aperture}
         \put(-32,80){\footnotesize  Optical beam}
            \end{overpic}  \vspace{-12pt}
              \caption{Optical beam, receiving aperture and the displacement vector $y_k$.}\label{fig-2}    
\end{figure}

The pointing error $y_k=d(\theta_k-\alpha_k)$ is a linear function of $x_{k}^p$ and $x_{k}^l$. It can be written as the following augmented system form
\begin{equation}\label{eq-proof1}
\left\{\begin{array}{llll}
x_{k+1}=\mathbb{A}x_k+\varphi(x_{k})+\mathbb{B}u_k+\mathbb{R}w_k,\\
\epsilon_k=\mathbb{C}x_k,
\end{array}\right.
\end{equation}
where
\begin{align*}
&\mathbb{A}\!=\!\!\begin{bmatrix} a^p & \mathbf{0} \\ \mathbf{0} & a^l\end{bmatrix}\!, \, \mathbb{B}\!=\!\begin{bmatrix} b^p \\ \mathbf{0}\end{bmatrix}\!,\, \mathbb{R}\!=\!\begin{bmatrix} r^p \\ r^l \end{bmatrix}\!, \, \mathbb{C}\!=\! \begin{bmatrix} r^p & -r^l \end{bmatrix},
\end{align*}
and $x_k\!=\!\begin{bmatrix} x_{{k}}^{{p}} \\ x_{{k}}^{{l}}  \end{bmatrix} \in \R^{2}$ is the augmented state vector, $w_k=\begin{bmatrix} w_{{k}}^{{p}} & w_{{k}}^{{l}}  \end{bmatrix}$ is the augmented disturbance vector, $y_k=d\epsilon_k$ is the pointing error and $\epsilon_k=\theta_k-\alpha_k$.

Since $\varphi(x_{k}^p)$ is  Lipschitz, then the augmented nonlinearity $\varphi(x_{k})$ is assumed to satisfy the following bound
\begin{equation}\label{condition1aa}
\varphi(x_{k})^T\varphi(x_{k})\leqslant  x_{k}^T\mathbb{H}^T\mathbb{H}x_{k},
\end{equation}
where $\mathbb{H}$ is a constant matrix.

The robust control problem studied in this paper consists in minimizing the closed-loop pointing error while ensuring the disturbance attenuation level. The objective is to maintain the centroid of the optical beam as close as possible to the center of the photodetector. This pointing problem can be interpreted as finding a set-point $u_k=-\mathbb{K}x_k$ depending on $x_k$ such that the following $\cHi$ norm of the pointing error $y_k$ with respect to disturbance $w_k$ is satisfied {\it i.e}: 
\begin{equation}\label{condition1}
\norm{y_k}\leqslant \varepsilon \norm{w_k},
\end{equation}
with $\varepsilon$ being the smallest positive real to be minimized. 

\section{Robust $\cHi$ Pointing Error Control for FSO}\label{main}
In this section, we consider the $\cHi$-norm optimization problem that guarantees the closed-loop pointing error $y_k$ is stable and ensures the disturbance attenuation level $\norm{y_k}\leqslant \varepsilon \norm{w_k}$ for a prescribed attenuation level $\varepsilon>0$.


The following theorem provides the stability and the absolute pointing error of the augmented system \eqref{eq-proof1}.
\begin{thm}\label{thm-1}
If there exist matrices $\mathbb{Y}=\mathbb{Y}^T>0$, $\mathbb{S}$ and scalar $\varepsilon$ such that the following LMI condition 
\begin{equation*}
\min_{\mathbb{S}, \mathbb{Y}>0}\,  \varepsilon \quad \mbox{subject to}
\end{equation*}
\begin{equation}\label{eq-thm-1}
\begin{bmatrix} -\mathbb{Y}&\mathbf{0} & \mathbf{0} &\mathbb{Y}\mathbb{A}^T\!-\!\mathbb{S}^T\mathbb{B} & \mathbb{Y}\mathbb{C}^T & \mathbb{Y}\mathbb{H}^T
\\ (\star) & -\varepsilon^2 \mathbf{I} & \mathbf{0} &\mathbb{R}^T & \mathbf{0} & \mathbf{0}
\\ (\star) & (\star) &-\delta \mathbf{I} &\mathbb{Y}& \mathbf{0} & \mathbf{0}
\\ (\star) & (\star) &(\star)  &-\mathbb{Y} & \mathbf{0}  & \mathbf{0}
\\ (\star) & (\star) &(\star)  &(\star) & -\mathbf{I}& \mathbf{0}
\\ (\star) & (\star) &(\star)  &(\star) & (\star) & -\delta^{-1}\!\mathbf{I}\end{bmatrix}\!<\!0.
\end{equation}
has a feasible solution with $\mathbb{K}=\mathbb{S}\mathbb{Y}^{-1}$, then
\begin{enumerate}
\def\labelenumi{\bf\theenumi)}\def\theenumi{\roman{enumi}}
  \item the augmented closed-loop error system \eqref{eq-proof1} is stable for $w_k=0$.
    \item and the pointing error $y_k$ satisfies the disturbance attenuation condition for $w_k\neq 0$ and a specific attenuation factor $\varepsilon>0$, 
  \begin{equation}\label{condition1bis}
\norm{y_k}\leqslant \varepsilon\norm{w_k}.
\end{equation}
\end{enumerate}
\end{thm}

\begin{proof}
Let us define the Lyapunov function $V_k=x_k^T\mathbb{P}x_k$ with $\mathbb{P}=\mathbb{P}^T>0$ and the $\cHi$ cost \citep{BEFB:94} in equation \eqref{condition1bis} as follows
\begin{equation}\label{cost1}
J\triangleq\sum_{k=0}^{N-1}\displaystyle\left(y_k^Ty_k -\varepsilon^2 w_k^Tw_k\right),
\end{equation}
over the time interval $[0, N\!-\!1]$, $N\!\in\!\N$. If $V_0=0$, {\it i.e}, all the initial conditions are null then, inequality \eqref{condition1bis} holds if the following satisfies along the trajectories of system \eqref{eq-proof1}
\begin{align}\label{eq-proof2a}
J<\sum_{k=0}^{N-1}\!\left(x^T_k\mathbb{C}^T\mathbb{C}x_k-\varepsilon^2w_k^Tw_k+V_{k+1}-V_k\right)\!.
\end{align}

A sufficient condition to fulfill the inequality \eqref{eq-proof2a} is to guarantee for all $k\in\mathbb{Z}_{\geqslant0}$
\begin{equation}\label{eq-proof3}
x^T_k\mathbb{C}^T\mathbb{C}x_k-\varepsilon^2w_k^Tw_k+V_{k+1}-V_k<0.
\end{equation}

Inequality \eqref{eq-proof3} can be written as follows
\begin{align}\label{eq-proof4}
&x^T_k\left[(\mathbb{A-BK})^T\mathbb{P}\mathbb{(A-BK)}\!-\!\mathbb{P}\right]x_k\!+\!x_k^T\mathbb{(A-BK)}^T\mathbb{P}\mathbb{R}w_k\!+\!x_k^T\mathbb{(A\!-\!BK)}^T\mathbb{P}\varphi(x_{k})\!+\!w^T_k\mathbb{R}^T\mathbb{P}\mathbb{(A\!-\!BK)}x_k\!
\notag\\&+\!w_k^T\mathbb{R}^T\mathbb{P}\mathbb{R}w_k\!+\!\varphi^T\!(x_{k})\mathbb{P}\mathbb{(A\!-\!BK)}x_k+\!\varphi^T\!(x_{k})\mathbb{P}\varphi(x_{k})\!+\!\varphi^T\!(x_{k})\mathbb{P}\mathbb{R}w_k+w_k^T\mathbb{R}^T\mathbb{P}\varphi(x_{k})+x^T_k\mathbb{C}^T\mathbb{C}x_k-\varepsilon^2w^T_kw_k\!<\!0.
\end{align}

Inequality \eqref{eq-proof4} can be expressed as
\begin{equation}\label{eq-proof9}
\displaystyle \begin{bmatrix} x_k\\ w_k\\ \varphi(x_{k})\end{bmatrix}^T \!\!\begin{bmatrix} \Xi & \mathbb{(A\!-\!BK)}^T \mathbb{P}\mathbb{R}& \mathbb{(A\!-\!BK)}^T \mathbb{P}\\
(\star) &  -\varepsilon^2\mathbf{I}\!+\!\mathbb{R}^T\mathbb{P}\mathbb{R} & \mathbb{R}^T\mathbb{P}
\\(\star) &  (\star) & \mathbb{P} \end{bmatrix}\begin{bmatrix} \!x_k\\w_k\\ \varphi(x_{k})\end{bmatrix}\!\!<\!0,
\end{equation}
where $\Xi=-\mathbb{P}+\mathbb{(A-BK)}^T \mathbb{P}\mathbb{(A-BK)}+\mathbb{C}^T\mathbb{C}$.

Using the fact that constraint \eqref{condition1aa} is equivalent to the following quadratic inequality for any $\delta>0$
\begin{equation}\label{eq-proof9aa}
\displaystyle \delta \begin{bmatrix} x_k\\ w_k\\ \varphi(x_{k})\end{bmatrix}^T \!\!\begin{bmatrix} \mathbb{H}^T\mathbb{H}&  \mathbf{0} &  \mathbf{0} \\
(\star) &   \mathbf{0} & \mathbf{0}
\\(\star) &  (\star) & -\mathbf{I}\end{bmatrix}\begin{bmatrix} \!x_k\\w_k\\ \varphi(x_{k})\end{bmatrix}\!\!\geqslant\!0,
\end{equation}
\noindent Combining together \eqref{eq-proof9} and \eqref{eq-proof9aa} gives
\begin{equation}\label{eq-proof9cc}
\displaystyle \begin{bmatrix} x_k\\ w_k\\ \varphi(x_{k})\end{bmatrix}^T \!\!\begin{bmatrix} \Xi_1 & \mathbb{(A\!-\!BK)}^T \mathbb{P}\mathbb{R}& \mathbb{(A\!-\!BK)}^T \mathbb{P}\\
(\star) &  -\varepsilon^2\mathbf{I}\!+\!\mathbb{R}^T\mathbb{P}\mathbb{R} & \mathbb{R}^T\mathbb{P}
\\(\star) &  (\star) & \mathbb{P}-\delta \mathbf{I} \end{bmatrix}\begin{bmatrix} \!x_k\\w_k\\ \varphi(x_{k})\end{bmatrix}\!\!<\!0,
\end{equation}
where $\Xi_1=-\mathbb{P}+\mathbb{(A-BK)}^T \mathbb{P}\mathbb{(A-BK)}+\mathbb{C}^T\mathbb{C}+\delta \mathbb{H}^T\mathbb{H}$.

Using the well-known Schur complement, we obtain inequality \eqref{eq-proof10}.
\begin{table*}[!t]
{\normalsize
\begin{center}
\line(1,0){480}
\end{center}
\begin{equation}\label{eq-proof10}
\displaystyle \begin{bmatrix} -\mathbb{P}\!+\!\mathbb{C}^T\mathbb{C}\!+\!\nu  \mathbb{H}^T\mathbb{H}& \mathbf{0}& \mathbf{0}\\ (\star) & -\varepsilon^2 \mathbf{I} & \mathbf{0}  \\(\star) & (\star) & -\nu \mathbf{I}\end{bmatrix}\!+\!\begin{bmatrix}  \mathbb{(A-BK)}^T\mathbb{P}\\  \mathbb{R}^T\mathbb{P} \\ \mathbb{P}\end{bmatrix}\mathbb{P}^{-1}\!\begin{bmatrix} \mathbb{P} \mathbb{(A-BK)}& \mathbb{P}\mathbb{R}  &\mathbb{P}\end{bmatrix}\!\!<\!0,
\end{equation}
\begin{center}
\line(1,0){480}
\end{center}}
\vspace{0.15cm}
\end{table*}

Applying again the Schur complement to \eqref{eq-proof10}, then we obtain the following sufficient condition
\begin{equation}\label{eq-proof16}
\begin{bmatrix} -\mathbb{P}&\mathbf{0} & \mathbf{0} &\mathbb{(A-BK)}^T\mathbb{P} & \mathbb{C}^T & \mathbb{H}^T
\\ (\star) & -\varepsilon^2\mathbf{I} & \mathbf{0} &\mathbb{R}^T\mathbb{P} & \mathbf{0} & \mathbf{0}
\\ (\star) & (\star) &-\delta \mathbf{I} &\mathbb{P} & \mathbf{0} & \mathbf{0}
\\ (\star) & (\star) &(\star)  &-\mathbb{P} & \mathbf{0}  & \mathbf{0}
\\ (\star) & (\star) &(\star)  &(\star) & -\mathbf{I}& \mathbf{0}
\\ (\star) & (\star) &(\star)  &(\star) & (\star) & -\delta^{-1}\!\mathbf{I}\end{bmatrix}\!<\!0.
\end{equation}
Pre- and post-multiplying inequality \eqref{eq-proof16} by\\ $\diag[\mathbb{Y}, \, \mathbf{I},\, \mathbf{I},\, \mathbb{Y},\, \mathbf{I},\, \mathbf{I}]$ and using $\mathbb{Y}\!=\!\mathbb{P}^{-1}$
and $\mathbb{S}\!=\!\mathbb{K}\mathbb{Y}$ which yields to the LMI \eqref{eq-thm-1}.

The fulfillment of inequality $\eqref{eq-thm-1}$ implies the fulfillment of the optimality condition:
\begin{equation}\label{eq-proof18}
\norm{y_k}\leqslant \varepsilon\norm{w_k}, \quad u_k\neq 0, \, w_k\neq 0.
\end{equation}
This completes the proof.
\end{proof}

\section{Numerical simulations}\label{results}
This section evaluates both the quality of the pointing error performance and the communication performance metrics to ensure accurate pointing and improve the optical communication link performance.

\subsection{Pointing Error Performance}\label{sub-sec1}
To evaluate the quality of the pointing error obtained by the LMI condition given in \eqref{eq-thm-1}. The YALMIP interface \citep{Lof:04} to MATLAB 8.5 with SDPT3 optimization toolbox \citep{TTT:99} was used to provide solutions. The performance of the pointing error control strategy is visualized through Fig. \ref{fig-Pointing} where $\varepsilon=0.32$ and $\mathbb{K}=\begin{bmatrix} 0.275 & -0.019\end{bmatrix}$. We observe clearly that the closed-loop pointing maintains a small alignment error between the optical beam transmitter and the receiving aperture system with a relatively error amplitude less than $\pm0.2$ {\si mm} from the detector diameter, which is $\pm0.5$ {\si mm}. Therefore the maximum relative error is $\pm40\%$.

\begin{figure}[!t]
\centering
      \begin{overpic}[scale=0.35]{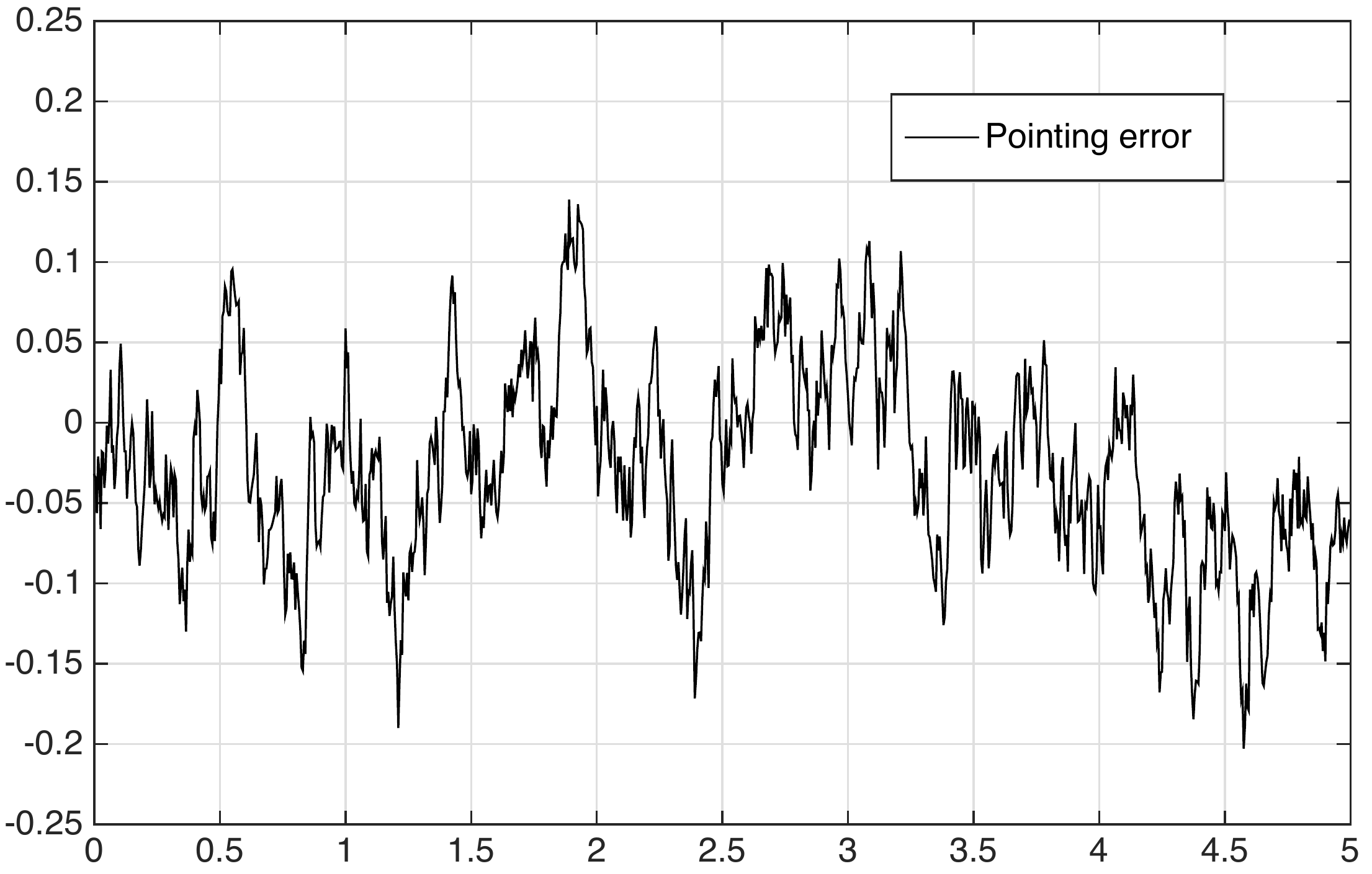}
         \put(-0.5,22){\scriptsize  \begin{rotate}{90} pointing error $y_k$ [{\si mm}]\end{rotate}}
 \put(44,-2){\scriptsize  Time [{\si sec}]}
            \end{overpic}  \vspace{-2pt}
              \caption{Closed-loop pointing error versus Time.}\label{fig-Pointing}    
\end{figure}

Statistical simulated values provide more analytical insight. Figs. \ref{fig-OO} and \ref{fig-CC} depict histogram data of the open-loop and closed-loop output-error displacements, respectively. The open-loop error displacement varies from $+2$ {\si mm} to $+7$ {\si mm}. The performance of open-loop with an output-error variance $\sigma^2\!=\!0.0380$ as expected due to bias and drift terms is inadequate and do not satisfactorily stabilize the beam at the center. For the closed-loop output-error feedback, the beam is stabilized at the center of the sensing device, the closed-loop error displacement varies from $-3$ {\si mm} to $+3$ {\si mm} and the output-error variance is $\sigma^2\!=\!0.0231$ which is $\pm 39\%$ reduced from the open-loop output error.
\begin{figure}[!t]
        \centering
     \begin{overpic}[scale=0.33]{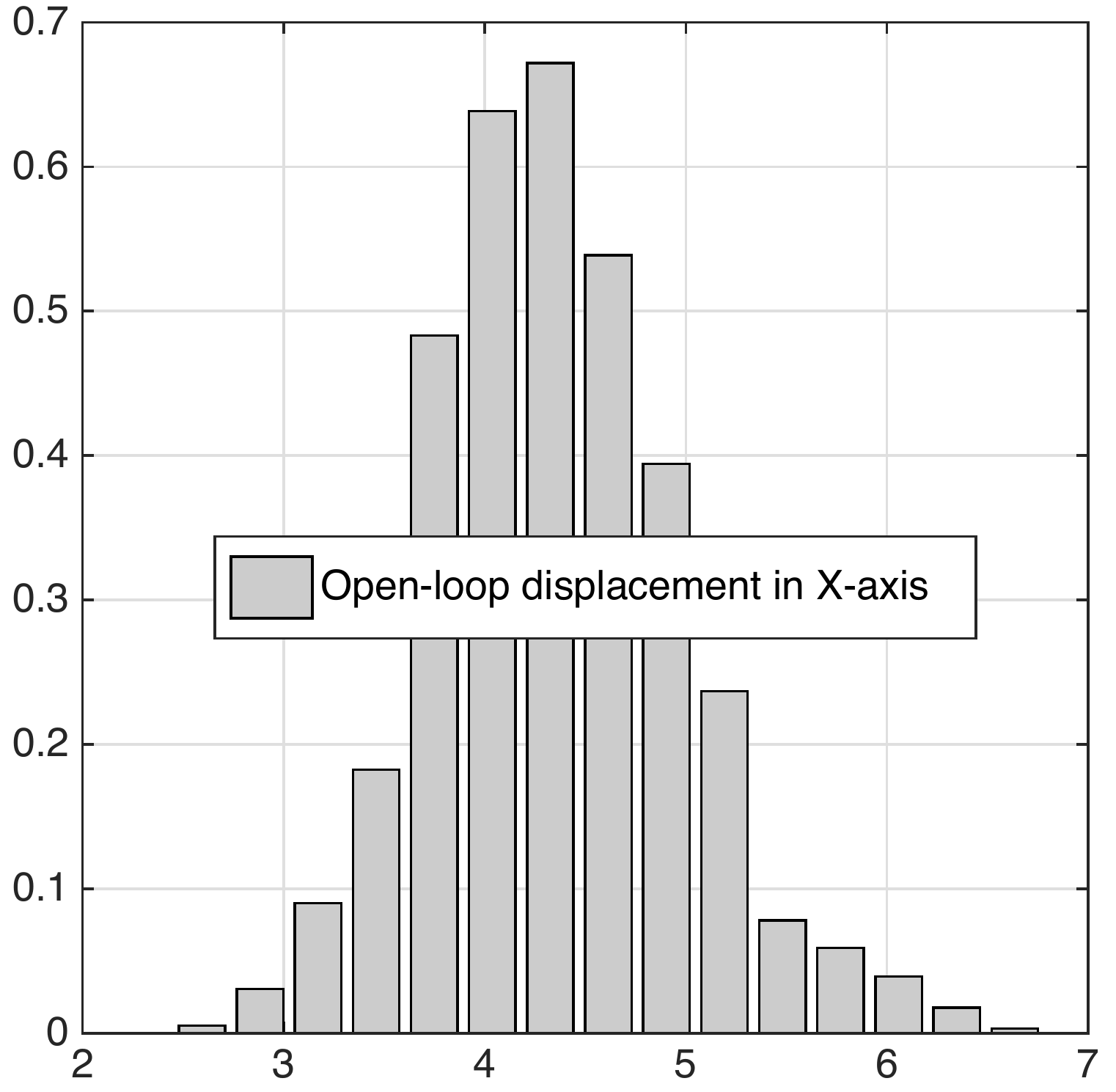}
   \put(-2,35){\scriptsize  \begin{rotate}{90} Probability density \end{rotate}}
 \put(39,-3){\scriptsize  Distance [mm]}
            \end{overpic}  \vspace{-2pt}
              \caption{Open-loop output-error displacement.}\label{fig-OO}
\end{figure}
\begin{figure}[!t]
        \centering
     \begin{overpic}[scale=0.33]{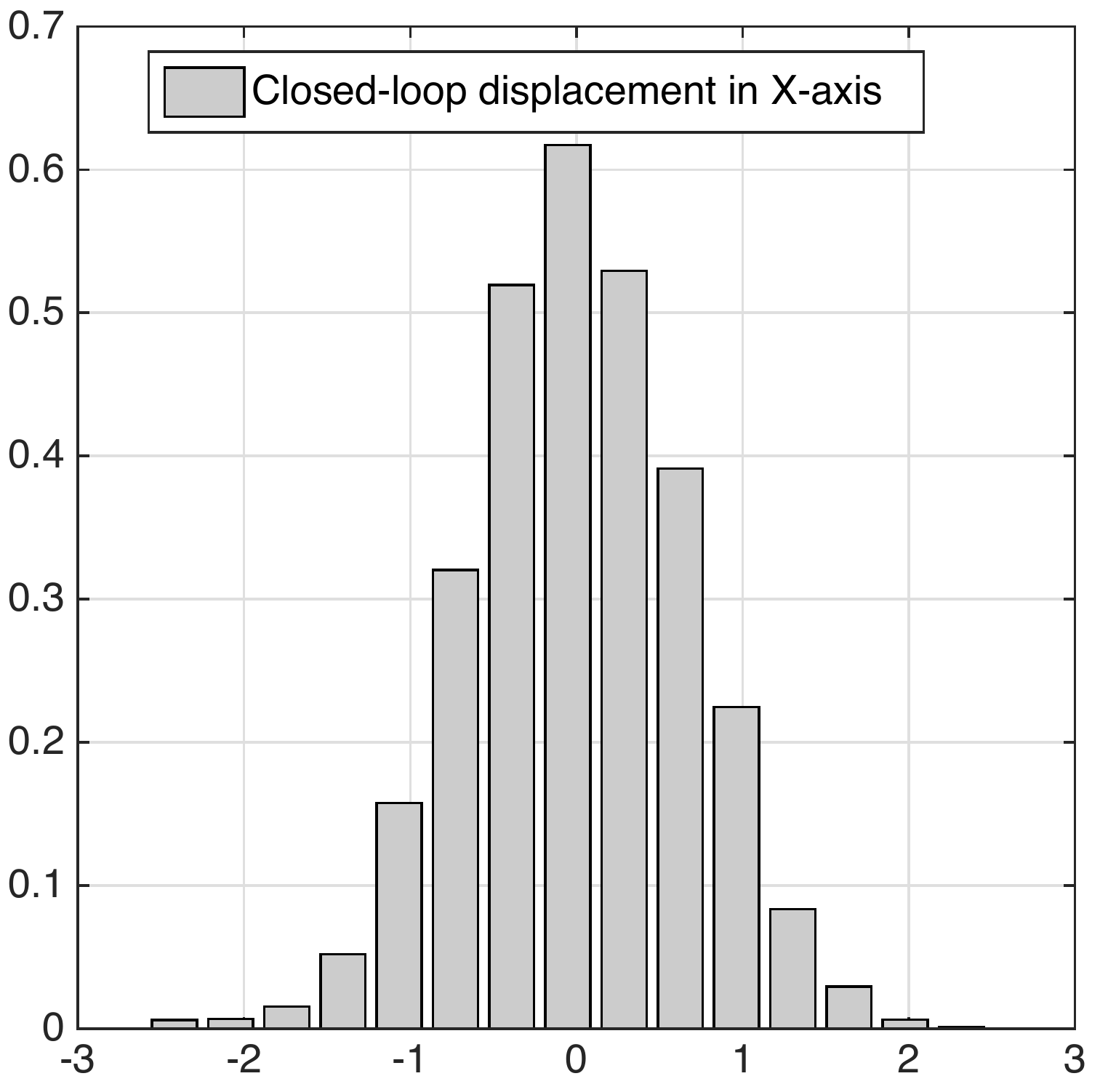}
   \put(-2,35){\scriptsize  \begin{rotate}{90} Probability density \end{rotate}}
 \put(38,-3){\scriptsize  Distance [mm]}
            \end{overpic}  \vspace{-2pt}
              \caption{Closed-loop output-error displacement.}\label{fig-CC}
\end{figure}

\subsection{Communication Performance Metrics}\label{sub-sec1a}
The outage probability error and bit error rate (BER) are metrics for quantifying communication systems' performance in fading channels. FSO system with a good average BER can temporarily suffer from increases in pointing error rate due to fading effects \citep{FaH:07,AlS:02,GPR:12,YCT:14}. The outage probability is given as follows \citep{GPR:12}
\begin{equation}\label{outage1}
P_{o}(I)=\!\!\!\int_0^{I_0/m}\!\!\!\!\!\frac{1}{\sqrt{2\pi\sigma^2}}\frac{1}{I}\displaystyle\exp\left\{-\frac{(\ln(I/I_0)+\sigma^2/2)^2}{2\sigma^2}\right\} \der \!I,
\end{equation}
where $m$ is the power margin which is introduced to account for the extra power needed to cater for turbulence-induced signal fading. Using Chernoff upper bound on \eqref{outage1}, an approximate power margin, $m$, needed to obtain $P_{o}$ can be obtained as follows \citep{GPR:12}
\begin{equation}\label{outage2}
m\approx\exp\left(\sqrt{-2\ln2 P_{o}\sigma^2} +\sigma^2/2\right).
\end{equation}

The evaluation of this outage probability error for the FSO link in open-loop and closed-loop conditions is depicted in Fig. \ref{fig-Outage}. The closed-loop outage probability error which output-error variance $\sigma^2\!=\!0.0231$ is reduced of $1.5$dBm from the open-loop outage probability error which output-error variance $\sigma^2\!=\!0.0380$. For example, to achieve an outage probability of $10^{-6}$ about $34.3$ {\si dBm} of extra power is needed in open-loop condition at $\sigma^2\!=\!0.0380$. This is reduced to $33.2$ {\si dBm} in closed-loop condition as the scintillation strength decreases to $\sigma^2\!=\!0.0231$.

The average BER of OOK-based FSO in atmospheric turbulence is defined as \citep{YCT:14}
\begin{equation}\label{ber1}
\mbox{BER}=\int_0^\infty p(I)Q\left(\frac{\eta I}{\sqrt{2N_0}}\right)\der \!I,
\end{equation}
where $\displaystyle Q(x)\!\!=\!\!\!\int_x^{+\infty}\!\!\!\exp(-t^2/2)\der\! t$, $\eta$ is the optical-to-electrical conversion coefficient,  $I$ represents the received optical intensity signal, $N_0$ is the additive white Gaussian noise power spectral density. The integration in \eqref{ber1} can be efficiently computed by Gauss-Hermite quadrature formula \citep{NUK:07,GPR:12,AlG:99,Osc:02,YCT:14}. Fig. \ref{fig-BER} shows the BER plots of OOK-based FSO in atmospheric turbulence corresponding to the open-loop and closed-loop data rate values at various levels of output-error variance. As we can see,  the effect of turbulence strength on the amount of signal-to-noise ratio (SNR) is required to maintain a given error performance level. From Fig. \ref{fig-BER}, it can be inferred that atmospheric turbulence can causes SNR penalty, which might affect the pointing error, for example, around $16$ {\si dB} of SNR is needed in open-loop condition to achieve a BER of $10^{-1}$ due to the very weak scintillation of strength $\sigma^2\!=\!0.0380$. It decreases by over $2$dB with the closed-loop control as the scintillation strength decreases to $\sigma^2\!=\!0.0231$, which implies that pointing error control strategy might be required to avoid a BER floor in the system~performance.     

\begin{figure}[!t]
        \centering
     \begin{overpic}[scale=0.39]{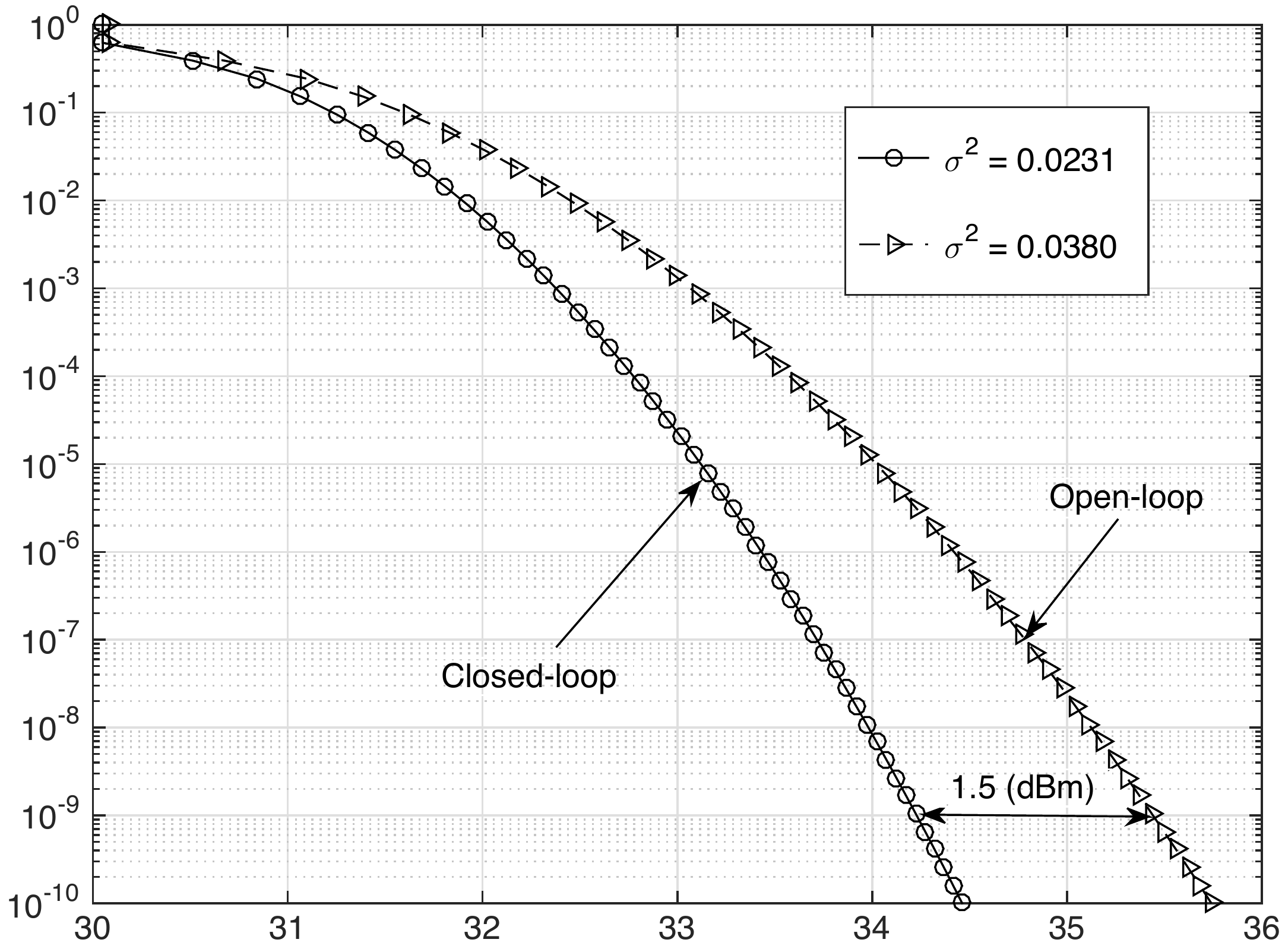}
             \put(-1,25){\footnotesize  \begin{rotate}{90} Outage Probability \end{rotate}}
 \put(40,-3){\footnotesize  Power Margin ({\si dBm})}
            \end{overpic}  \vspace{-2pt}
              \caption{Open-loop and closed-loop outage probability errors.}\label{fig-Outage}
\end{figure}
\begin{figure}[!t]
        \centering
     \begin{overpic}[scale=0.40]{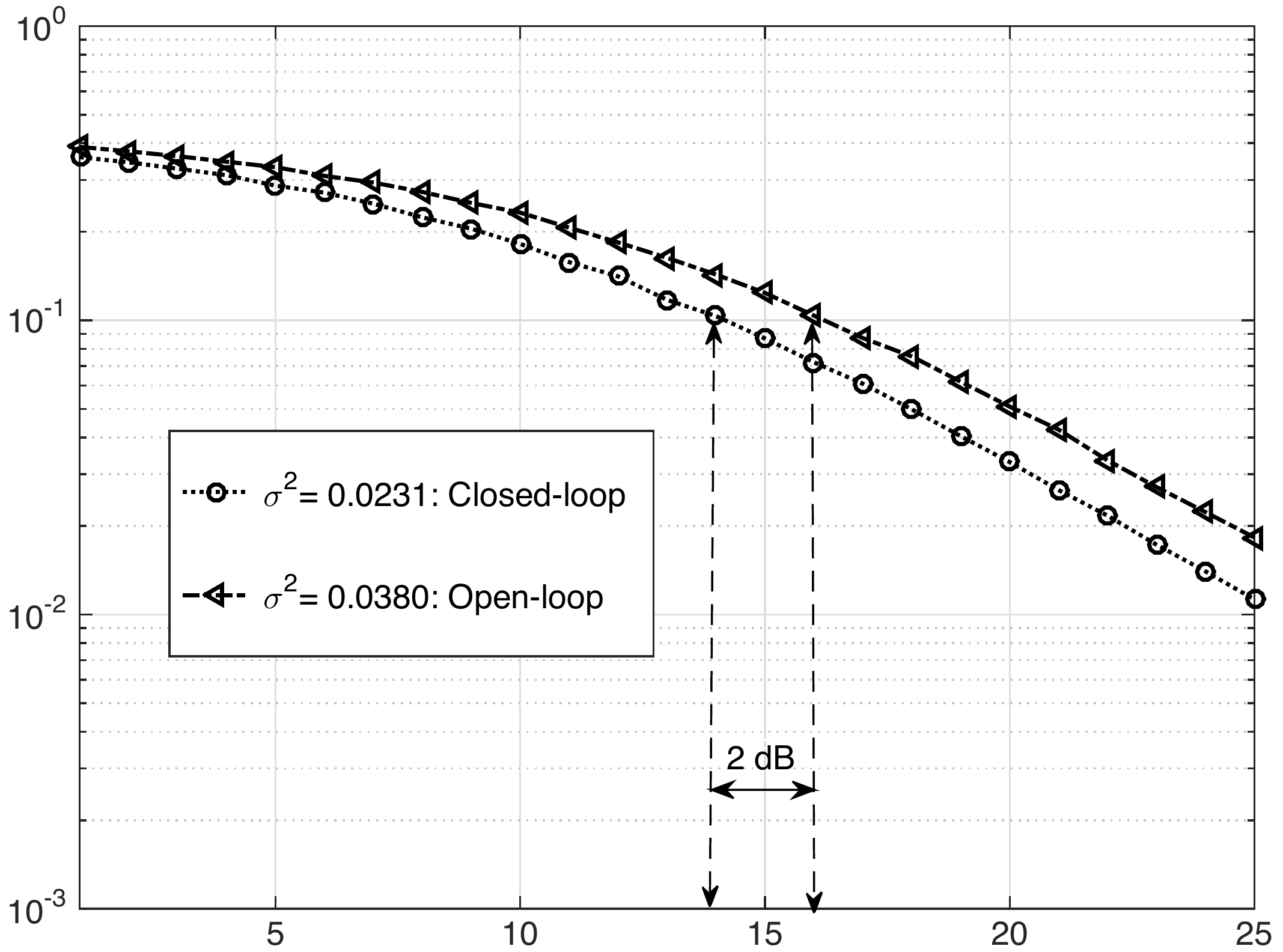}
        \put(-1,35){\footnotesize  \begin{rotate}{90} BER \end{rotate}}
 \put(46,-3){\footnotesize  SNR ({\si dB})}
            \end{overpic}  \vspace{-2pt}
              \caption{Open-loop and closed-loop BER performances of OOK-based FSO in atmospheric turbulence.}\label{fig-BER}
\end{figure}

\section{Conclusion}\label{conclusion}
In this paper, the link performance of the presented FSO link under the influence of the weak atmospheric turbulence has been investigated. The atmospheric turbulence chamber has been characterized theoretically and experimentally for a valid comparison. We found that the fading statistics follow the well-known lognormal distribution that is used for weak turbulence characterization. Based on that, a deterministic nonlinear discrete-time model for pointing error loss due to misalignment has been derived. We then investigate the $\cHi$ norm optimization problem that guarantees the closed-loop pointing error is stable and ensures the prescribed disturbance attenuation level. The closed-loop pointing error from the numerical simulation shows the center of the optical beam close enough to the receiving aperture center, verifying the efficiency of our proposed robust pointing error control for FSO communication systems.

\bibliography{MZ-1}

\end{document}